\documentclass[tightenlines,superscriptaddress,eqsecnum,floats,nofootinbib,showpacs]
{revtex4}
\usepackage{amssymb}
\usepackage{stmaryrd}
\usepackage{amsmath}
\usepackage{amsfonts}
\usepackage{mathrsfs}
\usepackage{amsmath,amssymb,amsfonts,amsthm}
\usepackage{graphicx}
\usepackage[T1]{fontenc}
\usepackage{xcolor}

\usepackage{hyperref} 

\newtheorem{thm}{Theorem}[section]

\newtheorem{lmm}{Lemma}[section]

\newcommand{\Pl}{\ell_\mathrm{p}} 

\newcommand{\sgn}{\mathrm{sgn}} 
\newcommand{\hil}{\mathcal{H}} 
\newcommand{\hmu}{\mathcal{H}_{\mathfrak b}}
\newcommand{\htau}{\mathcal{H}_{\mathfrak c}}
\newcommand{\dmu}{D_{\mathfrak b}}
\newcommand{\dd}{{\rm d}}
\newcommand{\h}{\mathfrak{h}}
\newcommand{\hbh}{\hil_{\rm BH}}
\renewcommand{\b}{\mathfrak b}
\renewcommand{\c}{\mathfrak c}

\newcommand{\hmum}{\hmu(m)}
\newcommand{\f}{\mathfrak{f}}

\makeatletter
\newsavebox{\@brx}
\newcommand{\llangle}[1][]{\savebox{\@brx}{\(\m@th{#1\langle}\)}%
  \mathopen{\copy\@brx\kern-0.5\wd\@brx\usebox{\@brx}}}
\newcommand{\rrangle}[1][]{\savebox{\@brx}{\(\m@th{#1\rangle}\)}%
  \mathclose{\copy\@brx\kern-0.5\wd\@brx\usebox{\@brx}}}
\newcommand{\llb}[1][]{\savebox{\@brx}{\(\m@th{#1(}\)}%
  \mathopen{\copy\@brx\kern-0.5\wd\@brx\usebox{\@brx}}}
\newcommand{\rrb}[1][]{\savebox{\@brx}{\(\m@th{#1)}\)}%
  \mathclose{\copy\@brx\kern-0.5\wd\@brx\usebox{\@brx}}}
\makeatother

\newcommand{\hpb}{\hat{p}_b}
\newcommand{\hpc}{\hat{p}_c}
\newcommand{\hp}{\hat{\h}}
\newcommand{\hpm}{\hp^{(m)}}

\begin{document}

\title{ Loop quantum deparametrized Schwarzschild interior and discrete black hole mass
 }

\author{Cong Zhang\footnote{ czhang@fuw.edu.pl}}
\affiliation{Faculty of Physics, University of Warsaw, Pasteura 5, 02-093 Warsaw, Poland}
\author{Yongge Ma\footnote{ mayg@bnu.edu.cn}}
\affiliation{Department of Physics, Beijing Normal University, Beijing 100875, China}

\author{Shupeng Song\footnote{songsp@bit.edu.cn}}
\affiliation{School of Physics, Beijing Institute of Technology, Beijing 100081, China}

\author{Xiangdong Zhang\footnote{ scxdzhang@scut.edu.cn}}
\affiliation{Department of Physics, South China University of
Technology, Guangzhou 510641, China}

\begin{abstract}
We present the detailed analyses of a model of loop quantum Schwarzschild interior coupled to a massless scalar field and extend the results in our previous rapid communication \cite{zhang2020loop} to more general schemes.  It is shown that the spectrum of the black hole mass is discrete and does not contain zero. This supports the existence of a black hole remnant after Hawking evaporation due to loop quantum gravity effects. Besides to show the existence of  a non-vanishing minimal black hole mass in the vacuum case, the quantum dynamics for the non-vacuum case is also solved and  compared with the effective one.
\end{abstract}
\pacs{04.60.Pp, 04.70.Dy}


\maketitle
\section{Introduction}
As a prediction of general relativity (GR), the existence of black holes has a broad base of support from observations \cite{Akiyama:2019cqa}.  However, our understanding on black hole (BH) is still far from the end. Among those challenging topics on BH, its quantum nature is particularly interesting. By studying quantum BHs, one could not only solve puzzles originating from the classical theory, but also achieve more understandings on the theory of quantum gravity. 

As a background independent approach to quantum gravity, loop quantum gravity (LQG) has been widely studied in the past 30 years \cite{ashtekar1991lectures,rovelli2004quantum,ashtekar2004background,han2007fundamental,thiemann2008modern,abhay2017loop}. 
Although some important breakthroughs have been made in LQG \cite{rovelli1995discreteness,ashtekar1995quantization,ashtekar1997quantum,ashtekar1997quantumII,thiemann1998quantum,thiemann1998quantumII,ma2000qhat,lewandowski2005uniqueness,yang2016new}, its dynamics is still an open issue. The obstacle of LQG can be bypassed through applying the loop quantization techniques to the symmetry-reduced sectors of GR, where the  expression of the Hamiltonian constraint becomes much simpler than that in the full theory. The resulting quantum models are expected to reflect some quantum features of full LQG, in spite of the fact that they might not be equivalent to the direct symmetric sector of full LQG. An improved treatment of quantum-reduced LQG has also been proposed to study the symmetric sectors \cite{alesci2013quantumreduced}. These ideas were applied to study loop quantum Schwarzschild BH recently with different perspectives \cite{ashtekar2005quantum,modesto2006loop,boehmer2007loop,campiglia2008loop,chiou2008phenomenological,gambini2013loop,corichi2016loop,bojowald2018effective,ashtekar2018quantum,ashtekar2018quantumII,rovelli2018small,bianchi2018white,achour2018polymer,alesci2018quantum,alesci2019quantum,bodendorfer2019mass,bodendorfer2019effective,kelly2020black,han2020improved}. However, most of these studies focused on the effective dynamics, where one considered the Hamiltonian constraint with the holonomy correction and solved the effective Hamiton's equations.  This treatment resulted in several important achievements.  In particular, it resolves the singularity inside the Schwarzschild BH
and predicts certain extensions of the Schwarzschild interior beyond the singularity (see, for instance, \cite{ashtekar2018quantum,ashtekar2018quantumII,bodendorfer2019effective,kelly2020black,han2020improved}).  However, in the effective prescription one cannot see more intrinsic quantum natures of BH, such as the ground state of quantum BHs and the discreteness of the spectrum of Dirac observables. After all, it is necessary to consider the quantum dynamics. 

There are several crucial topics on the quantum dynamics of BH. One is the issue of the final state of BH evaporation which is related to the  constituent of dark matter and the puzzle of information loss.
According to the Hawking radiation \cite{hawking1974black}, the primordial mini BHs in the very early universe should be completely evaporated by now.  However, if the BH evaporation is halted at some stable state by some quantum gravity effect, which is called the BH remnant, these remnants would result in important cosmological consequences \cite{barrow1992cosmology,carr1994black,dalianis2019primordial}. Remarkably, the remnants originating from these primordial BHs could even comprise the entire dark matter in the universe \cite{barrow1992cosmology,carr1994black}. Moreover, thanks to the remnant, one could argue that the information fallen into a BH with matters could be stored in the remnant after its evaporation. This  provides a possible approach to solve the puzzle of information loss \cite{preskill1992black,chen2015black}. Furthermore, the distortion of the semiclassical Hawking spectrum resulted from certain discreteness of the BH mass  was studied \cite{bekenstein1997quantum,pranzetti2012radiation,barrau2015black,lochan2016discrete}. 
It was argued that in certain cases, the distortion could be observable even for macroscopic BHs \cite{bekenstein1997quantum}. Although these debates are crucial and long standing, there was no systematic study by quantum gravity to lay a solid theoretical foundation for the arguments until the prediction of a BH remnant in LQG models by \cite{zhang2020loop}. The purpose of this paper is to provide the detailed constructions in \cite{zhang2020loop} and extend the results to more general schemes. Moreover, the quantum dynamics of the model will be further studied in detail.

We study the model of loop quantum Schwarzschild interior coupled to a massless scalar field. The quantum Schwarzschild BH, as the vacuum case of this model, can be resulted by vanishing the scalar field. 
The phase space of this system contains three pairs of canonical variables, $(b,p_b)$ and $(c,p_c)$ for gravity, and $(\varphi,p_\varphi)$ for the scalar field. By deparametrizing this model, one gets the physical Hamiltonian $\sqrt{\h}$  of the relational evolution with respect to the scalar field $\varphi$ \cite{alesci2015hamiltonian,dapor2013relational,zhang2020quantum}. In the classical theory, the Poisson bracket between $\h$ and the ADM mass $M$ vanishes. This indicates a classical Dirac observable $m=c p_c$ proportional to $M$.  
However, this commutativity may no longer be kept by the corresponding operators $\hp$ and $\hat m$ in the quantum theory, 
which is relevant to the choice of schemes for the quantization of $\h$. 
We only focus on the schemes such that   $\hp$ and $\hat m$ are still commutative, since the commutativity means the existence of a Dirac observable $\hat m$. 
Note that 
a general class of schemes adopted for the loop quantization of the current model can meet our requirement. In particular, it is valid for the $\mu_o$-scheme \cite{ashtekar2005quantum,modesto2006loop} and the new scheme balancing the $\mu_o$ scheme and $\bar\mu$ scheme \cite{corichi2016loop,ashtekar2018quantum,ashtekar2018quantumII}. However, it cannot be met by the $\bar\mu$ scheme \cite{boehmer2007loop,chiou2008phenomenological}. 

We will first construct the Hamiltonian operator $\hp$ and study its properties analytically. Thanks to these analytical results, a  numerical method to diagonalize $\hp$ is proposed so that the dynamics is computable. 
Then the quantum dynamics of the model can be solved for the non-vacuum case, and it can be compared with the effective one. For the vacuum case, 
the Hilbert space consisting of the physical states of the Schwarzschild BH
is built up. In this Hilbert space $m$ is promoted to an operator  which has discrete spectrum $\overline{\sigma_\xi}$ with $0\notin \overline{\sigma_\xi}$. 
This result supports the existence of a stable BH remnant. 

This paper is organized as follows. In Sec \ref{section2}, the theory of loop quantum Schwarzschild BH interior coupled to a massless scalar field is briefly reviewed, including the deparametrization and the polymer quantization of this model.  In Sec. \ref{section4}, we construct the physical Hamiltonian operator and study its properties analytically. Then the quantum dynamics for both non-vacuum and vacuum cases are solved in Sec. \ref{section6}. Finally, in Sec. \ref{sec:conclusion}, our results are summarized and discussed. 

\section{Preliminaries}\label{section2}
\subsection{Deparametrization of the model}\label{sec:deparametrization}
Given a spatially homogeneous  3-manifold $\Sigma$ of topology $\mathbb R\times S^2$. Because of the homogeneity, $\Sigma$ is endowed with a fiducial metric
 \begin{equation}\label{eq:fiducialmetric}
\mathring{q}_{ab}\dd x^a\dd x^b=\dd x^2+r_o^2(\dd \theta^2+\sin^2\theta\dd\phi^2)
\end{equation}
where $(x,\theta,\phi)$ are the natural coordinates adapted to the topology, and $r_o$ is a constant with dimension of length. Since $\Sigma$ is non-compact in the $x$ direction, we introduce an elementary cell  $\mathcal{C}\cong (0,L_0)\times \mathbb{S}^2$ in $\Sigma$ and restrict all integrals to this elemental cell to avoid the divergence of integrations. 

The classical phase space of gravity coupled to a massless scalar field contains the Ashtekar-Barbero canonical conjugate pairs $(A_a^i(x),E^a_i(x))$ for gravity and $(\varphi(x),\pi(x))$ for scalar field. As far as the homogeneous states are concerned, the scalar filed $\varphi$ is reduced to a constant and the fields $A_a^i(x)$, $E^a_i(x)$ and $\pi_\varphi(x)$ take the forms \cite{zhang2020quantum}
\begin{equation}\label{eq:reducedAEPi}
\begin{aligned}
A_a^i\tau_i\dd x^a&=\frac{c}{L_0}\tau_3\dd x+b\tau_2\dd\theta-b\tau_1\sin\theta\dd\phi+\tau_3\cos\theta\dd\phi\\
E_i^a\tau^i\partial_a&=p_c\tau_3\sin\theta\partial_x+\frac{p_b}{L_0}\tau_2\sin\theta\partial_\theta-\frac{p_b}{L_0}\tau_1\partial_\phi\\
\pi_\varphi&=\frac{p_\varphi}{4\pi r_0^2}\sqrt{\mathring{q}}=\frac{p_\varphi}{4\pi }\sin\theta
\end{aligned}
\end{equation}
where $\tau_j=-i\sigma_j/2$ $(j=1,2,3)$ with $\sigma_j$ being the Pauli matrix and, $c,\ b,\ p_c,\ p_b$ and $p_\varphi$ are all constants. 
According to \eqref{eq:reducedAEPi}, the symmetry-reduced phase space is coordinatized by the pairs $(c,p_c)$, $(b,p_b)$ and $(\varphi,p_\varphi)$. The non-vanishing Poisson brackets read
\begin{equation}\label{eq:clapoisson}
\{c,p_c\}=2G\gamma,\ \{b,p_b\}=G\gamma,\ \{\varphi,p_\varphi\}=1.
\end{equation}
By the symmetry-reduced expression \eqref{eq:reducedAEPi}, the Gaussian and diffeomorphism constraints vanish automatically. 
The dynamics of this model is encoded in the Hamiltonian constraint. In the full theory of gravity coupled to a massless scalar field, the Hamiltonian constraint can be deparametrized as \cite{domagala2010gravity}
\begin{equation}\label{eq:Hphysical}
\begin{aligned}
C(x)=\pi_\varphi(x)\pm \sqrt{h(x)}=0,
\end{aligned}
\end{equation}
where $h(x)=-2\sqrt{|\det(E(x))|}C_{\rm gr}(x)$, with $C_{\rm gr}$ being the vacuum-gravity Hamiltonian constraint. The so-called physical Hamiltonian can be written as
\begin{equation}\label{eq:Hphysical1}
h_{\rm phy}=\frac{\sqrt{G}\gamma}{\sqrt{4\pi}}\int_{\Sigma}\dd^3 x\sqrt{h(x)}=\frac{\gamma}{4\sqrt{2}\pi}\int_{\Sigma}\dd^3 x\sqrt{E_i^aE_j^b\left(-F^k_{ab}\epsilon_{ijk}+2(1+\gamma^2)K_{[a}^iK_{b]}^j\right)}
\end{equation}
where $F^k_{ab}$ denotes the curvature of the connection $A_a^i$, $K_a^i$ is the extrinsic curvature and
the prefactor is adapted for convenience. It generates the relational evolution with respect to the scalar field. 
Substituting \eqref{eq:reducedAEPi} in to \eqref{eq:Hphysical} and integrating both sides in the elementary cell $\mathcal C$, one finally obtains  the Hamiltonian constraint of our model as
\begin{equation}\label{eq:classicalreducedH}
p_\varphi\pm \frac{\sqrt{4\pi}}{\sqrt{G} L_0\gamma}\sqrt{\h}=0
\end{equation}
where $\h$ is given by
\begin{equation}\label{eq:clh0}
\h:= p_b  \left(\left(b^2+\gamma ^2\right)p_b +2 b c p_c\right)
\end{equation}
In this work, we will also refer to the physical Hamiltonian as $\h$ though it is related to the true physical Hamiltonian $h_{\rm phy}$ by $\sqrt{\h}=h_{\rm phy}$. 

The quantum states of the model are described by vectors in the Hilbert space $\mathcal H_t=\mathcal H_{\rm mat}\otimes \mathcal H_{\rm gr}$, where $\mathcal H_{\rm mat}$ and $\mathcal H_{\rm gr}$ are the Hilbert spaces describing the matter and gravity respectively. The physical states $|\psi\rangle_{\rm phy}$ should satisfy the quantum version of the Hamiltonian constraint, i.e.,
\begin{equation}\label{eq:quantumreducedH}
\left(\hat p_\varphi \pm \frac{\sqrt{4\pi}}{L_0\gamma\sqrt{G}} \widehat{\sqrt{\h}}\right)|\psi\rangle_{\rm phy}=0.
\end{equation} 
It gives us a Sch\"odinger-like equation. Therefore, in the case that the operator $\widehat{\sqrt{\h}}$ is self-adjoint, the physical states, i.e., the solutions of \eqref{eq:quantumreducedH}, can be expressed as
\begin{equation}\label{eq:phystate}
|\psi\rangle_{\rm phy}=e^{\mp i\frac{\sqrt{4\pi G}}{L_0\gamma \Pl^2}\varphi \widehat{\sqrt{\h}}}|\psi\rangle_{\rm gr}
\end{equation}
where $|\psi\rangle_{\rm gr}\in\mathcal H_{\rm gr}$ is a state of pure gravity  and $\Pl:=\sqrt{G\hbar}$ is the Plank length. A Dirac observable $\mathcal O_{\rm phy}$ takes the form
\begin{equation}
\mathcal O_{\rm phy}=e^{\mp i\frac{\sqrt{4\pi G}}{L_0\gamma\Pl^2}\varphi  \widehat{\sqrt{\h}}}\mathcal O_{\rm gr}e^{\pm i\frac{\sqrt{4\pi G}}{L_0\gamma\Pl^2}\varphi  \widehat{\sqrt{\h}}}
\end{equation} 
with an operator $\mathcal O_{\rm gr}$ in $\mathcal H_{\rm gr}$. 

To carry out the above deparametrization procedure, in next subsection we will introduce the polymer quantization for the gravity and obtain its Hilbert space $\mathcal H$. Then a self-adjoint operator $\hp$ on $\mathcal{H}$ is proposed by the loop quantization procedure. Finally we restrict ourselves to the subspace consisting of non-negative spectra of $\hp$ to define  the Hilbert space
\begin{equation}\label{eq:hgr}
\mathcal H_{\rm gr}:=\hat P_{[0,\infty)}\mathcal H
\end{equation} 
where $\hat P_{[0,\infty)}:=\chi_{[0,\infty)}(\hp)$ is the projection operator with respect to the spectrum decomposition of $\hp$. 

\subsection{The polymer quantization }\label{sec:polymerquantization}
The polymer quantization of the gravity in this model leads to the Hilbert space \cite{ashtekar2005quantum}
\begin{equation}
\tilde{\mathcal H}=\tilde{\mathcal H}_{\b}\otimes\tilde{\mathcal H}_{\c}=L^2(\mathbb{R}_{\rm Bohr},\dd\mu_0)\otimes L^2(\mathbb{R}_{\rm Bohr},\dd\mu_0),
\end{equation}
where $\dd\mu_0$ is the Haar measure on the Bohr compactification $\mathbb{R}_{\rm Bohr}$ of the real line $\mathbb{R}$ \cite[Chapter 28]{thiemann2008modern}. The two spaces $\tilde{\mathcal H}_{\b}$ and $\tilde{\mathcal H}_{\c}$ correspond to the canonical conjugate pairs $(b,p_b)$ and $(c,p_c)$ respectively. The standard basis of these two Hilbert spaces are denoted by $|\mu\rangle\in\tilde{\mathcal H}_{\b}$ and $|\tau\rangle\in\tilde{\mathcal H}_{\c}$. Their inner products read
\begin{equation}\label{eq:basis}
\langle \mu'|\mu\rangle=\delta_{\mu',\mu},\ \langle\tau'|\tau\rangle=\delta_{\tau',\tau},
\end{equation}
where the Kronecker-$\delta$ symbol is employed. 

There are two types of basic operators in the Hilbert spaces. One is the momentum-variable operators $\hpb$ and $\hpc$, whose actions on $|\mu\rangle$ and $|\tau\rangle$ are given by
\begin{equation}
\begin{aligned}
\hpb\, |\mu\rangle&=\frac{\gamma \Pl^2}{2}\mu|\mu\rangle\\
\hpc\, |\tau\rangle&=\gamma \Pl^2\,\tau|\tau\rangle.
\end{aligned}
\end{equation}
The other is the configuration-variable operators $\widehat{e^{i\lambda b}}$ and $\widehat{e^{i\lambda c}}$, whose actions read 
\begin{equation}
\begin{aligned}
\widehat{e^{i\lambda b}}|\mu\rangle&=|\mu+2\lambda\rangle\\
\widehat{e^{i\lambda c}}|\tau\rangle&=|\tau+2\lambda\rangle.
\end{aligned}
\end{equation}
As in the model of loop quantum cosmology (LQC), the operators $\widehat{e^{i\lambda b}}$ and $\widehat{e^{i\lambda c}}$  correspond to the holonomies along the edges parallel to  the $\mathbb R$ direction and the equator (or the  longitude because of the homogeneity) of $\mathbb S^2$ respectively. 
  Moreover, $\widehat{e^{i\lambda b}}$ and $\widehat{e^{i\lambda c}}$ are not strongly continuous with respect to $\lambda$. Therefore there are no operators corresponding to the configuration variables $b$ and $c$ in our model. This respect to the fact that there does not exist operator corresponding to the connection itself in the full theory. 

\section{The physical Hamiltonian operator}\label{section4}
\subsection{A separable subspace $\mathcal H$ of $\tilde{\mathcal H}$}\label{sec:hilbertspaces}
According to \eqref{eq:basis} , the Hilbert space $\tilde{\mathcal H}$ possesses a non-countable orthonormal basis, and hence it is non-separable. The problem that the kinematical Hilbert space is non-separable has appeared in both LQG and LQC. In LQG, the non-separability is caused by the uncountability of the graphs based on which the Hilbert space is defined. In LQC, it is caused by the polymer quantization procedure which leads to an uncountable orthonormal basis as the case of the present model.  In LQG, one proposed to employ the diffeomorphism invariance to identify the diffeomorphism equivalent graphs in order to obtain a separable diffeomorphism invaraint Hilbert space \cite{ashtekar2004background}.
In LQC, the problem is tackled by the super-selection feature of the Hamiltonian constraint. More precisely, because the Hamiltonian constraint operator preserves some  separable Hilbert subspaces, one can confine the study in certain separable subspace \cite{ashtekar2003mathematical,ashtekar2006quantumnature}. The situation in the present model is very similar to that in LQC. As shown below, a separable subspace $\hil\subset\tilde\hil$ can be selected, which is preserved by the Hamiltonian operator. 

Classically, the physical Hamiltonian \eqref{eq:clh0}, re-denoted by  $\h_c$, can be rewritten as
\begin{equation}\label{eq:hamiltoniancl}
\h_c=2 p_b b  p_cc+p_b^2 b^2+\gamma ^2p_b^2,
\end{equation}
where $p_c c$ is a Dirac observable. Actually, for vacuum gravity, one has $p_c c=L_0\gamma G M$ with $M$ being the ADM mass of Schwarzschild BH \cite{zhang2020quantum}. Since in the Hilbert space $\tilde{\mathcal H}$ there are no operators corresponding to $b$ and $c$, the expression \eqref{eq:hamiltoniancl} cannot be promoted to an operator directly. We thus return to the integral expression \eqref{eq:Hphysical1}  of the full deparametrized theory and follow \cite{ashtekar2005quantum} to express $C_{\rm gr}$ in terms of $\mathring{F}=\dd (\gamma K)+[\gamma K,\gamma K]$  and the spatial curvature $\Omega=-\sin(\theta)\dd\theta\wedge\dd\phi\tau_3$, with 
\begin{equation}\label{eq:expressK}
K=\frac{1}{\gamma}\left(\frac{c}{L_0}\tau_3\dd x+b\tau_2\dd\theta-b\tau_1\sin\theta\dd\phi\right)
\end{equation}
denoting the extrinsic curvature.
Then the physical Hamiltonian $h_{\rm phy}$ is expressed by
\begin{equation}\label{eq:integralexpression}
h_{\rm phy}=\frac{1}{4\sqrt{2}\pi}\int_{\mathcal C} \dd^3 x\sqrt{E^a_i(x)E^b_j(x)(\mathring{F}^k_{ab}(x)-\gamma^2 \Omega_{ab}^k(x))}.
\end{equation}
Note that, in comparison with the expression \eqref{eq:reducedAEPi} of $A$, the expression \eqref{eq:expressK} does not contain the term $\tau_3\cos(\theta)\dd\phi$. Thus  the ``holonomy'' of $K$ along the $\phi$-direction is much simpler than that of $A$. Moreover, since $\Omega$ does not depend on dynamical variables, one needs not to regularize it by holomolies.
To regularize \eqref{eq:integralexpression}, one fixes three edges intersecting at a point $x_0\in \mathcal C$, where
 the first edge $e_1$ is along the $\mathbb R$ direction of $\Sigma$ taking length $\ell_1=\delta_c L_0$ and the other two edges $e_2$ and $e_3$ are along the equator and the longitude of $\mathbb S^2$ with the same radians $\ell_2=\ell_3=\delta_b$. By defining the ``holonomies'' 
 \begin{equation}
 h_i=\exp\left(\int_{e_i}K\right)
 \end{equation}
 one can regularize $\mathring{F}_{ab}^k$ as \cite{ashtekar2005quantum}
  \begin{equation}\label{eq:regularizedF0}
 \mathring{F}_{ab}^k(x_0)=-\sum_{i,j}\frac{2}{\ell_i\ell_j}\mathrm{tr}\left(h_i h_jh_i^{-1}h_j^{-1}\tau_k\right)\mathring{\omega}_a^i\mathring{\omega}_b^j+O(\sqrt{\ell_a\ell_b})
 \end{equation}
 where $\mathring{\omega}_a^i$ denotes the co-triad of the fiducial metric \eqref{eq:fiducialmetric} , adapting to the  three edges $e_i$. Substituting \eqref{eq:regularizedF0} and the expression \eqref{eq:reducedAEPi} of  $E^a_i$ into \eqref{eq:integralexpression}, we thus  get a regularized version of $h_{\rm phy}=\sqrt{\h}$ for our model, where
\begin{equation}\label{eq:clah}
\h=2 p_b\frac{ \sin(\delta_bb)}{\delta_b}  p_c\frac{\sin(\delta_c c)}{\delta_c}+p_b^2 \frac{\sin^2(\delta_b b)}{\delta_b^2}+\gamma ^2p_b^2.
\end{equation}
  The expression \eqref{eq:clah} returns to $\h_c$ in \eqref{eq:hamiltoniancl} at the limit of $\delta_b\to 0$ and $\delta_c\to 0$.

To quantize the regularized physical Hamiltonian \eqref{eq:clah} instead of \eqref{eq:hamiltoniancl}, 
various strategies have been proposed in choosing the two quantum parameters $\delta_b$ and $\delta_c$. Roughly speaking, the choices can be classified into three schemes. The first one is the $\mu_0$-scheme where $\delta_b$ and $\delta_c$ are chosen as constants \cite{ashtekar2005quantum,modesto2006loop,campiglia2008loop}. The second one is the $\bar \mu$-scheme which allows $\delta_b$ and $\delta_c$ to be any functions of $p_b$ and $p_c$ \cite{boehmer2007loop,chiou2008phenomenological}.  The third one, referred to as the modified scheme, is developed recently where $\delta_b$ and $\delta_c$ are phase space dependent only through Dirac observables \cite{corichi2016loop,ashtekar2018quantum,ashtekar2018quantumII}. Since the expression $p_c\sin(\delta_c c)/\delta_c$ in \eqref{eq:clah} corresponds to the classical Dirac observable $p_c c=L_0\gamma GM$ in \eqref{eq:hamiltoniancl}, we are motivated to assume that the quantum operator $\hat\h$ is composed of a Dirac observable $\widehat{p_c\sin(\delta_c c)/\delta_c}$, which is self-adjoint and commutes with $\hat \h$. 
This assumption rules out certain strategies such as the $\bar\mu$-scheme where $\delta_c$ depends on $\hat{p}_b$ and $\delta_b$ depends on $\hat p_c$, since the resulting operator would no longer commute with $\hat\h$. 
Furthermore, we assume that $\delta_b$ is a constant or any function of the Dirac observable $\widehat{p_c\sin(\delta_c c)/\delta_c}$ for following analysis.
Similar assumptions were adopted in $\mu_o$-scheme and the modified scheme. They are sufficient but not necessary to obtain a Dirac observable $\widehat{p_c\sin(\delta_c c)/\delta_c}$. In summary, the current paper focuses on the schemes such that 
\begin{itemize}
\item[(i)] a separable Hilbert subspace $\htau\subset\tilde{\mathcal H}_{\c}$ can be chosen to define an operator $\widehat{p_c\sin(\delta_c c)/\delta_c}$ corresponding to $p_c\sin(\delta_c c)/\delta_c$ which is self-adjoint, commutates with the physical Hamiltonian $\hat \h$. 
\item[(ii)] the quantum parameter $\delta_b$ is a constant or any function of $p_c\sin(\delta_c c)/\delta_c$.
\end{itemize}

 Denote $\sigma_{\mathfrak c}$ as the spectrum of the Dirac observable $\widehat{p_c\sin(\delta_c c)/\delta_c}$. Then $\htau$ is isometric to the Hilbert space $L^2(\sigma_{\mathfrak c},\dd\mu_{\mathfrak c}))$ with the spectral measure $\dd \mu_{\mathfrak c}$. The Hilbert space $\tilde\hmu\otimes\htau\subset\tilde\hil$ can be represented by $ L^2(\sigma_{\mathfrak c},\dd\mu_{\mathfrak c}; \tilde{\mathcal H}_{\mathfrak b})$ in which each element is a $\tilde{\mathcal H}_{\mathfrak b}$-valued function on $\sigma_{\mathfrak c}$. The representation is defined by 
\begin{equation}\label{eq:identity}
U:\tilde{\mathcal H}_{\mathfrak b}\otimes L^2(\sigma_{\mathfrak c},\dd\mu_{\mathfrak c})\ni \psi_b\otimes\psi_c\mapsto \psi_c(\cdot)\psi_b\in L^2(\sigma_{\mathfrak c},\dd\mu_{\mathfrak c};\tilde{\mathcal H}_{\mathfrak b}).
\end{equation}
The inner product in $L^2(\sigma_{\mathfrak c},\dd\mu;\tilde{\mathcal H}_{\mathfrak b})$ is given by
\begin{equation}\label{eq:innerH}
\left(\psi^{(1)},\psi^{(2)}\right)=\int_{\sigma_{\mathfrak c}}\dd \mu_{\mathfrak c}\langle \psi^{(1)}(x)|\psi^{(2)}(x)\rangle
\end{equation}
where $\langle \psi^{(1)}(x)|\psi^{(2)}(x)\rangle$ denotes the inner produce of $\psi^{(1)}(x),\psi^{(2)}(x)\in\tilde{\mathcal H}_{\mathfrak b}$. 

For convenience, the elements in the spectrum space can be denoted by $L_o\gamma m\in \sigma_{\mathfrak c}$ with $m\in \mathbb R$ in analogy with their classical correspondence $p_c c=L_o\gamma GM$. We also define
\begin{equation}
\hat\beta_{\lambda}:=\widehat{p_b\sin(\lambda b)}.
\end{equation}
For a state $\psi\in L^2(\sigma_{\mathfrak c},\dd\mu_{\mathfrak c}; \tilde{\mathcal H}_{\mathfrak b})$, $\psi(L_o\gamma m)\in \tilde{\mathcal H}_{\mathfrak b}$ is abbreviated to $\psi(m)$.
Then the action of $\hat \h$ on $\psi$ reads
\begin{equation}\label{eq:actionofh}
(\hat \h\psi)(m)=\left(\frac{2L_o\gamma m}{\delta_b^{(m)}}\hat\beta_{\delta_b^{(m)}}+\frac{1}{(\delta_b^{(m)})^2}\hat\beta_{\delta_b^{(m)}}^2+\gamma^2\hat p_b^2\right)\psi(m)
\end{equation}
where $\delta_b^{(m)}\equiv \delta_b(L_o\gamma m)$ due to the dependence of $\delta_b$  on $\widehat{p_c\sin(\delta_c c)/\delta_c}$. 
By \eqref{eq:actionofh}, the separable Hilbert subspace $\mathcal H\subset L^2(\sigma_{\mathfrak c},\dd\mu_{\mathfrak c}; \tilde{\mathcal H}_{\mathfrak b})$ preserved by $\hat\h$ is constructed as follows. Given $L_o\gamma m\in \sigma_{\mathfrak c}$, let $\hmu(m)\subset\tilde{\mathcal H}_{\mathfrak b}$ be a separable Hilbert space preserved by $\hat \beta_{\delta_b^{(m)}}$. Then the Hilbert space $\hil\subset L^2(\sigma_{\mathfrak c},\dd\mu_{\mathfrak c};\tilde{\mathcal H}_\b)$ contains the state $\psi$ such that $\psi(m)\in\hmum$ for all $L_o\gamma m\in\sigma_{\mathfrak c}$. For convenience, the Hilbert space $\mathcal H$ constructed by $\hmu(\cdot):m\mapsto \hmum$ through this procedure will be denoted by $L^2(\sigma_{\mathfrak c},\dd\mu_{\mathfrak c};\hmu(\cdot))$.

\subsection{The operators $\hat \beta_{\lambda}$}\label{sec:betaoperator}
According to \eqref{eq:actionofh}, the action of $\hat \h$ is determined by the property of $\hat \beta_{\delta_b^{(m)}}$. Since $\delta_b^{(m)}$ is constant for a given $m\in \mathbb R$, we will drop the superscript in $\delta_b^{(m)}$ and consider $\hat\beta_{\delta_b}$ with a constant $\delta_b$. 
Classically, one has
\begin{equation}\label{eq:classicalpb}
\begin{aligned}
p_b\sin(\delta_b b)&=\frac{1}{2i}\left(p_b e^{i\delta_b b}-p_b e^{-i\delta_b b}\right).
\end{aligned}
\end{equation}
A well-known ambiguity in the quantization procedure is caused by the operator ordering. By definition, one has
\begin{equation}
\widehat{e^{i\delta_b b}}\hat p_b=(\hat p_b-\delta_b \gamma\Pl^2)\widehat{e^{i\lambda b}}.
\end{equation}
We hereby introduce a parameter $\b$ to parameterize various operator-ordering strategies and define
\begin{equation}
\widehat{p_be^{i\delta_b b}}=(\hat p_b+ \gamma\Pl^2\b)\widehat{e^{i\delta_b b}}. 
\end{equation}
Then, we can define a symmetric operator corresponding to \eqref{eq:classicalpb} as
\begin{equation}\label{eq:oppsin}
\begin{aligned}
\hat \beta_{\delta_b}&:=\frac{1}{2i}\left((\hpb+\gamma\Pl^2\b)\widehat{e^{i\delta_b b}}-\widehat{e^{-i\delta_b b}}(\hpb+\gamma\Pl^2\b)\right).
\end{aligned}
\end{equation}
 Its action is given by 
\begin{equation}\label{eq:psin}
\begin{aligned}
\hat \beta_{\delta_b}|\mu\rangle&=\frac{1}{2i}\frac{\gamma\Pl^2}{2}\Big( (\mu+2\delta_b+2\b)|\mu+2\delta_b\rangle-(\mu+2\b)|\mu-2\delta_b\rangle\Big).
\end{aligned}
\end{equation}
Thus the separable Hilbert subspaces of $\tilde{\mathcal H}_{\b}$ preserved by $\hat\beta_{\delta_b}$   are given by
\begin{equation}\label{eq:hbc}
\begin{aligned}
\hmu^{(\varepsilon_b)}&=\overline{\{\psi\in\tilde{\mathcal H}_{\b},\psi(\mu)\neq 0 \text{ only for }\mu=\varepsilon_b+2 n\delta_b, n\in\mathbb Z \}}
\end{aligned}
\end{equation}
for some constant $\varepsilon_b\in [0,2\delta_b)$. We will show in below that the separable Hilbert subspace preserved by the physical Hamiltonian $\hat \h$ can be constructed with  $\hmu^{\varepsilon_b}$. 

Denote the restriction of $\hat\beta_{\delta_b}$ on $\hmu^{(\varepsilon_b)}$ by $\hat\beta_{\delta_b}\restriction\hmu^{(\varepsilon_b)}$. 
Given $\hmu^{\varepsilon_b}$ and $\hmu^{\tilde \varepsilon_b}$, a natural isomorphism between them can be defined by 
\begin{equation}\label{eq:identityep}
i:\hmu^{(\varepsilon_b)}\ni |\varepsilon_b+2 n\delta_b\rangle\mapsto |\tilde \varepsilon_b+2 n\delta_b\rangle\in \hmu^{(\tilde \varepsilon_b)}.
\end{equation}
It introduces an operator $i(\hat \beta_{\delta_b}\restriction\hmu^{(\varepsilon_b)} )i^{-1}$ in $\hmu^{(\varepsilon_b)}$ whose action reads
\begin{equation}\label{eq:ibetai}
i(\hat \beta_{\delta_b}\restriction\hmu^{(\varepsilon_b)})i^{-1}|\tilde\varepsilon_b+2 n\delta_b\rangle=\frac{1}{2i}\frac{\gamma\Pl^2}{2}\Big( (2n\delta_b+\tilde\varepsilon_b+2\delta_b+2\tilde\b)|\tilde\varepsilon_b+2 n\delta_b+2\delta_b\rangle-(2n\delta_b+\tilde\varepsilon_b+2\tilde\b)|\tilde\varepsilon_b+2 n\delta_b-2\delta_b\rangle\Big),
\end{equation}
where $\tilde \b:=\b+(\varepsilon_b-\tilde\varepsilon_b)/2$. For clarity, let us use $\hat \beta_{\delta_b}^{(x)}$  to denote the operator $\hat \beta_{\delta_b}$ with respect to the constant $\b=x$. According to \eqref{eq:ibetai}, to study the operator $\hat\beta_{\delta_b}^{(x)}\restriction\hmu^{(\varepsilon_b)}$ with respect to $x\neq 0$, one can use the operator $i(\hat \beta_{\delta_b}^{(x)}\restriction\hmu^{(\varepsilon_b)} )i^{-1}$ on $\hmu^{\tilde\varepsilon_b}$ with $\tilde\varepsilon_b=\varepsilon_b+2x$, and thus has 
\begin{equation}\label{eq:ibetabe0}
i(\hat \beta_{\delta_b}^{(x)}\restriction\hmu^{(\varepsilon_b)} )i^{-1}=\hat \beta_{\delta_b}^{(0)}\restriction\hmu^{(\tilde\varepsilon_b)}.
\end{equation}
Therefore, without loss of generality, we can set $\b=0$ and study properties of $\hat \beta_{\delta_b}\restriction\hmu^{(\varepsilon_b)}$  for any given $\varepsilon_b$.

From now on, let us refer to the restriction $\hat \beta_{\delta_b}\restriction\hmu^{(\varepsilon_b)}$ for some $\hmu^{(\varepsilon_b)}$ as $\hat \beta_{\delta_b}$ unless being specially noted. Because $\hat \beta_{\delta_b}$ is unbounded, one has to assign certain domain to complete its definition. A natural choice of the domain $D(\hat\beta_{\delta_b})$ reads
\begin{equation}\label{eq:domainbeta}
D(\hat\beta_{\delta_b})=\{\psi\in\hmu^{(\varepsilon_b)}, |\mathrm{supp}(\psi)|<\infty\}
\end{equation}
where $|\mathrm{supp}(\psi)|$ denotes the cardinality of the support of $\psi$. This implies that $D(\hat\beta_{\delta_b})$ only consists of finite linear combination of the basis $|\mu\rangle \in \hmu^{(\varepsilon_b)}$. 
The essential self-adjointness of $\hat\beta_{\delta_b}$ with the domain $D(\hat\beta_{\delta_b})$ can be proven. Define a self-adjoint operator $\hat{N}(\geq 1)$ as $\hat N|\mu\rangle=(1+\mu^2)|\mu\rangle$. Then it is straight-forward to verify that there exist numbers $c,d\in\mathbb R$ such that
\begin{equation}
\begin{aligned}
\|\hat\beta_{\delta_b}\psi\|&\leq c\|N\psi\|, \forall \psi\in D\\
\left|\langle\hat\beta_{\delta_b}\psi,\hat N\psi\rangle-\langle\hat N\psi,\hat\beta_{\delta_b}\psi\rangle \right|&\leq d\|\hat N^{1/2}\psi\|^2,\ \forall \psi\in D.
\end{aligned}
\end{equation}
where $\langle\cdot,\cdot\rangle$ denotes the inner product in $\hmu^{(\varepsilon_b)}$. Therefore $\hat\beta_{\delta_b}$ is essentially self-adjoint with the domain $D(\hat\beta_{\delta_b})$ according to Theorem X.37 in\cite{reed2003methods}. 

By expanding the eigen-equation $\hat\beta_{\delta_b}|\psi\rangle=\omega|\psi\rangle$ with the basis $|\mu\rangle$, one has
\begin{equation}
\omega \psi(\mu)=\frac{1}{2i}\frac{\gamma\Pl^2}{2}\Big( (\mu+2\delta_b)\psi(\mu+2\delta_b)-\mu\psi(\mu-2\delta_b)\Big).
\end{equation}
As shown in appendix \ref{ap:assyptoticbeta}, $\psi(\mu)$ behaves asymptotically as
\begin{equation}\label{eq:asspbeta}
\psi(\mu)=\left\{
\begin{aligned}
&\frac{1}{\sqrt{|\mu|}}\left(\chi_+e^{ik\ln(|\mu|)}+\chi_-e^{-ik\ln(|\mu|)}\right)+O(|\mu|^{-1}),\ \mu=\varepsilon_b+4n\delta_b\\
&\frac{1}{\sqrt{|\mu|}}\left(\chi_+e^{ik\ln(|\mu|)}-\chi_-e^{-ik\ln(|\mu|)}\right)+O(|\mu|^{-1}),\ \mu=\varepsilon_b+(4n+2)\delta_b
\end{aligned}
\right.
\end{equation}
where $k=\frac{\omega}{\gamma\Pl^2\delta_b}$ and $n\in\mathbb Z$. This implies that
$\int|\psi(\mu)|^2\dd\mu=\infty.$ and hence $$\sum_{\mu=\varepsilon_b+2n}|\psi(\mu)|^2=\infty.$$
Thus $\psi(\mu)$ is unnormalizable. 
Moreover, since $\psi(\mu)$ behaves asymptoticly as a plane wave of $\ln(|\mu|)$, we can show by using the same techniques as in \cite{zhang2018towards} and the Weyl's criterion \cite{reed2003methods},  that the spectrum of $\widehat{p_b\sin(\delta_b b)}$ is the entire real line $\mathbb R$.
Furthermore, the two leading-order functions in \eqref{eq:asspbeta}, denoted by
\begin{equation*}
\psi_0^\pm(\mu):=\frac{1}{\sqrt{|\mu|}}\left(\chi_+e^{ik\ln(|\mu|)}\pm \chi_-e^{-ik\ln(|\mu|)}\right),
\end{equation*}
satisfy
\begin{equation}
-i\gamma\Pl^2\delta_b\sgn(\mu)\sqrt{|\mu|}\frac{\dd}{\dd\mu}\sqrt{|\mu|}\psi_0^\pm(\mu)= \omega \psi_0^\mp(\mu),
\end{equation}
where $\sgn(\mu)$ denotes the sign function of $\mu$. This implies that the operator $\hat\beta_{\delta_b}$,  in the large $\mu$ limit, returns to the Sch\"odinger quantization of $ p_b b$,
\begin{equation}
 \widehat{p_b b}=-i\gamma\Pl^2\sgn(\mu)\sqrt{|\mu|}\frac{\dd}{\dd\mu}\sqrt{|\mu|}.
\end{equation}
 Therefore, the classical limit of $\hat \beta_{\delta_b}$ is correct. This finishes the quantization procedure of $p_b\sin(\delta_b b)$ and the study of the properties of the corresponding operator $\hat \beta_{\delta_b}$. The same discussion can be transported analogously to $p_c\sin(\delta_c c)$ for constant $\delta_c$.  The resulting operator  $\widehat{p_c\sin(\delta_c c)/\delta_c}$ is self-adjoint and possesses the entire line as its spectrum. These properties of $\widehat{p_c\sin(\delta_c c)/\delta_c}$ are compatible with our general assumptions.

\subsection{The operator $\hp$}
Let us denote the operator in the RHS of \eqref{eq:actionofh} as
\begin{equation}\label{eq:321}
\hp^{(m)}:=\frac{2L_o\gamma m}{\delta_b}\hat\beta_{\delta_b}+\frac{1}{\delta_b^2}\hat\beta_{\delta_b}^2+\gamma^2\hat p_b^2.
\end{equation}
It should be noted that the operator $\hat\beta_{\delta_b}$ in \eqref{eq:321} is the original one defined by \eqref{eq:psin} with non-vanishing $\b$. However, 
as discussed below \eqref{eq:ibetai}, the identity map $i$ from $\hmu^{(\varepsilon_b)}$  to the specific Hilbert space $\hmu^{(\tilde\varepsilon_b)}$ with $\tilde\varepsilon_b=\varepsilon_b+2\b$ ensures that we can set $\b=0$. Thus, by this treatment we define an operator in $\hmu^{(\tilde\varepsilon_b)}$ corresponding to $\hp^{(m)}$ as
\begin{equation}\label{eq:ihmi000}
i(\hp^{(m)}\restriction \hmu^{(\varepsilon_b)}) i^{-1}=\frac{2L_o\gamma m}{\delta_b}i(\hat\beta_{\delta_b}\restriction \hmu^{(\varepsilon_b)}) i^{-1}+\frac{1}{\delta_b^2} i(\hat\beta_{\delta_b}^2\restriction \hmu^{(\varepsilon_b)}) i^{-1}+\gamma^2i(\hat p_b^2\restriction \hmu^{(\varepsilon_b)}) i^{-1}.
\end{equation}
where the map $i$, defined by \eqref{eq:identityep}, identifies the Hilbert spaces $\hmu^{(\varepsilon_b)}$ and $\hmu^{(\tilde\varepsilon_b)}$ with $\tilde\varepsilon_b=\varepsilon_b+2\b$. For the first two terms in this operator, one has 
\begin{equation}
i((\hat\beta_{\delta_b})^n\restriction \hmu^{(\varepsilon_b)}) i^{-1}=(\hat\beta_{\delta_b}^{(0)})^n\restriction \hmu^{(\tilde \varepsilon_b)}.
\end{equation}
where $\hat\beta_{\delta_b}^{(0)}$ is defined by \eqref{eq:ibetabe0}. For the last term, we have
\begin{equation}\label{eq:ipb2i}
i(\hat p_b^2\restriction \hmu^{(\varepsilon_b)}) i^{-1}|\tilde\varepsilon_b+2n\delta_b\rangle=[(\hat p_b-\gamma\Pl^2\b)^2\restriction \hmu^{(\tilde\varepsilon_b)}]|\tilde\varepsilon_b+2n\delta_b\rangle.
\end{equation}
Therefore, \eqref{eq:ihmi000} can be expressed as 
\begin{equation}\label{eq:hmrestricted}
i(\hp^{(m)}\restriction \hmu^{(\varepsilon_b)}) i^{-1}=\frac{2L_o\gamma m}{\delta_b}(\hat\beta_{\delta_b}^{(0)}\restriction \hmu^{(\tilde\varepsilon_b)}) +\frac{1}{\delta_b^2} [(\hat\beta_{\delta_b}^{(0)})^2\restriction \hmu^{(\tilde\varepsilon_b)}]+\gamma^2[(\hat p_b-\gamma\Pl^2\b)^2\restriction \hmu^{(\tilde\varepsilon_b)}].
\end{equation}
That is, by setting $\b=0$ for the $\hat\beta_{\delta_b}$ in $\hpm$, $\hat p_b$ has to be replaced by $\hat p_b-\gamma\Pl^2\b$. Thus we re-define $\hpm$ as
\begin{equation}\label{eq:hpmp}
\begin{aligned}
\hp^{(m)}=&\frac{2L_o\gamma m}{\delta_b}\hat\beta_{\delta_b}+\frac{1}{\delta_b^2} \hat\beta_{\delta_b}^2+\gamma^2(\hat p_b-\gamma\Pl^2\b)^2,
\end{aligned}
\end{equation}
where  $\hat\beta_{\delta_b}$ is defined by \eqref{eq:psin} with $\b=0$. As a consequence, the operator $\hp$ is changed correspondingly to
\begin{equation}
(\hp\psi)(m)=\hp^{(m)}\psi(m)
\end{equation}  
with the re-defined $\hp^{(m)}$.

Now let us construct the separable Hilbert $L^2(\sigma_\c,\dd\mu_\c;\hmu(\cdot))$ as mentioned below \eqref{eq:actionofh}. To do this, we only need to assign to each $m$ the Hilbert space $\hmu^{(\varepsilon_b(m))}$, i.e., to define $\hmu(\cdot)$ as
\begin{equation}\label{eq:hbm0}
\hmu(\cdot):m\mapsto \hmum:=\hmu^{(\varepsilon_b(m))},
\end{equation} 
where $\varepsilon_b:m\mapsto \varepsilon_b(m)\in [0,2\delta_b)$ is assumed to be some sufficiently well-behaved function. Then the resulting Hilbert space $L^2(\sigma_\c,\dd\mu_\c;\hmu(\cdot))$ is supposed to be acted by $\hp$. To define the domain of the unbounded operator $\hp$, we first define the domain $\dmu(m)$ of $\hp^{(m)}\restriction \hmu^{(\varepsilon_b(m))}$ for each $m$ by
\begin{equation}\label{eq:Dbm}
\dmu(m):=\{\psi\in \hmu^{(\varepsilon_b(m))},\langle \mu|\psi\rangle\neq 0 \text{ for finite numbers of } \mu\}.
\end{equation}
Then the domain of $\hp$, denoted by $D(\hp)$, reads
\begin{equation}
D(\hp)=\{\psi\in L^2(\sigma_\c,\dd\mu_\c;\hmu(\cdot)),\psi(m)\in \dmu(m)\ \forall m\}.
\end{equation}
We now complete the rigorous definition of the operator $\hp$. Since $\hp$ is identical to the operator-valued function $\hp^{(\cdot)}:m\mapsto\hpm$, instead of investigating $\hp$ itself, it is sufficient to study $\hp^{(m)}$ for all $m$. 
In the following, to simplify our notions, we will abbreviate $\varepsilon_b(m)$ to $\varepsilon_b$ if the dependence of $\varepsilon_b$ on $m$ does not matter. Moreover, the restricted operator  $\hp^{(m)}\restriction \hmu^{(\varepsilon_b)}$ with the domain $\dmu(m)$ is denoted simply by $\hp^{(m)}$.  Furthermore, when $\hp$ is mentioned, it refers to the operator $\hp$ with domain $D(\hp)$.

The operator $\hp$ and $\hp^{(m)}$ are necessary to be self-adjoint to govern a well-defined dynamics. By using the Kato-Rellich Theorem \cite{reed2003methodsii}, we can prove that both $\hp$ and $\hpm$ are essentially self-adjoint (see Appendix \ref{app:selfadjoint} for details). Because of the essential self-adjointness of $\hp$ and $\hp^{(m)}$, their closure, denoted by $\overline{\hpm}$ and $\overline{\hp}$ respectively, can be regarded as the Hamiltonian operator of the current model with the desired properties. We use $\overline{\dmu(m)}$ and $\overline{D(\hp)}$ to denote the domains of  $\overline{\hpm}$ and $\overline{\hp}$ respectively. 

\section{Dynamics govern by $\overline{\hp}$}\label{se:propertiesh}
 
To solve the dynamics governed by $\overline{\hp}$, one needs to diagonalize $\overline{\hp}$, or equivalently, to diagonalize $\overline{\hp^{(m)}}$ for all $m$. To begin with, we first study the discreteness of the spectrum of $\overline{\hp^{(m)}}$.  
By definition, $\hpm$ can be re-expressed as
\begin{equation}
\hpm=-L_0^2m^2\gamma^2+\left(L_0m\gamma+\frac{1}{\delta_b}\hat \beta_{\delta_b} \right)^2+\gamma^2(\hat p_b-\gamma\Pl^2\b)^2.
\end{equation}
Thus, one has
\begin{equation}\label{eq:hpmg}
\langle\psi|\overline{\hpm}|\psi\rangle\geq \gamma^2\langle\psi|(\hat p_b-\gamma\Pl^2\b)^2| \psi\rangle-L_0^2m^2\gamma^2\langle\psi|\psi\rangle,\ \forall |\psi\rangle\in \overline{\dmu(m)}.
\end{equation}
Hence $\overline{\hpm}$ is bounded from below. By \eqref{eq:hpmg}, we can obtain
\begin{equation}\label{eq:lowerbounded}
\mu_n(\overline{\hpm})\geq -L_0^2m^2\gamma^2+\gamma^2\mu_n((\hat p_b-\gamma\Pl^2\b)^2)
\end{equation} 
where, for an operator $\hat A$, $\mu_n(\hat A)$ denotes
\begin{equation}
\mu_n(\hat A):=\sup_{\varphi_1,\cdots,\varphi_{n-1}\in \hmu^{\varepsilon_b}}\inf_{\substack{\psi\in \overline{\dmu(m)};\|\psi\|=1;\\ \langle\varphi_i,\psi\rangle=0,\forall i=1,\cdots,n-1}}\langle \psi,\hat A\psi\rangle.
\end{equation}
Note that  $(\hat p_b-\gamma\Pl^2\b)^2$ can be defined on $\overline{\dmu(m)}$ by definition.  Hence $\mu_n((\hat p_b-\gamma\Pl^2\b)^2)$ is well-defined.
Since $\mu_n((\hat p_b-\gamma\Pl^2\b)^2)\to \infty$ as $n\to \infty$, one has that 
\begin{equation}
\lim_{n\to \infty}\mu_n(\overline{\hpm})=\infty.
\end{equation}
Hence, according to the min-max principle (see, e.g., Theorem XIII.1 in \cite{reed2003methodsiv}), $\overline{\hpm}$ has purely discrete spectrum. In other words, each element in the spectrum of $\overline{\hpm}$, denoted by $\sigma(\overline{\hpm})$,  is an eigenvalue of $\overline{\hpm}$ with finite multiplicity. 

Given the significance of $\overline{\hp^{(m)}}$, it is desirable to understand the properties of $\sigma(\hpm)$. In particular, one may ask whether the eigenvalue $\omega(m)$ as a function of $m$ is analytic or not. This issue is closely related to the analyticity of $\overline{\hpm}$ on $m$ in the sense of Kato \cite{reed2003methodsiv,kato2013perturbation}.
To overcome the technical difficulty that  $\overline{\hp^{(m)}}$ for different $m$ are defined in different Hilbert spaces,  we employ the following unitary map for a given $m_o$,
\begin{equation}
\mathfrak i_{m}: \hmu^{(\varepsilon_b(m))}\ni |\varepsilon_b(m)+2n\delta_b^{(m)}\rangle\mapsto |\varepsilon_b(m_o)+2n\delta_b^{(m_o)}\rangle\in \hmu^{(\varepsilon_b(m_o))}
\end{equation}
where  the dependence of $\varepsilon_b$ and $\delta_b$ on $m$ is written explicitly. The issue on the analyticity of $\omega(m)$ can be equivalently discussed by that of $\mathfrak i_{m}\overline{\hp^{(m)}}\mathfrak i_{m}^{-1}$ which is defined on $\mathfrak i_{m}\overline{\dmu(m)}=\overline{\dmu(m_o)}$. For simplicity, let us use $\delta_b^o$ and  $\varepsilon_b^o$ to denote $\delta_b^{(m_o)}$ and $\varepsilon_b(m_o)$ respectively.  

An  sesquilinear form $T_m(\cdot,\cdot)$ associated to $\mathfrak i_{m}\hp^{(m)}\mathfrak i_{m}^{-1}$ can be defined on $\dmu(m_o)$ by
\begin{equation}
T_m(\psi,\phi)=\langle\psi|\mathfrak i_{m}\hp^{(m)}\mathfrak i_{m}^{-1}|\phi\rangle,\ \psi,\phi\in \dmu(m_o).
\end{equation}
Then \eqref{eq:hpmg} implies that $T_m$ is semibounded, i.e.,
\begin{equation}
T_m(\psi,\psi)\geq -L_0^2m^2\gamma^2\|\psi\|^2,\ \forall \psi\in \dmu(m_o). 
\end{equation}
Hence, $T_m$ is closable. Its closure, denoted by $\overline{T_m}$, is the extension of $T_m$ on the closure of $\dmu(m_o)$ with respect to the norm
\begin{equation}
\|\psi\|_{m}=\sqrt{T_m(\psi,\psi)+(L_0^2m^2\gamma^2+1)\langle\psi|\psi\rangle}.
\end{equation} 
Since the norm $\|\psi\|_{m}$ depends on $m$, the resulting closures $Q(m)$ of $\dmu(m_o)$ could depend on $m$ in general. Suppose that $\delta_b^{(m)}$ and $\varepsilon_b(m)$ are analytic functions of $m$ at $m_o$ and $\delta_b^{(m_o)}\neq 0$.
It turns out that $\overline{T_m}$ is an analytic family of type (a) in the sense of Kato (see section VII.4 in \cite{kato2013perturbation}). This, by definition, is ensured by that i) $Q(m)=Q(m_o)$ for all $m$ sufficiently close to $m_o$, and ii) $\overline{T_m}(\psi,\psi)$ is analytic on $m$ for all $\psi\in Q(m_o)$. The detailed proof of this conclusion is presented in Appendix \ref{app:analyticity}. 

Because $\overline{T_m}$ is symmetric and closed, there is a self-adjoint operator $\hat{\mathfrak t}_m$ associated to it, which is indeed the Friedrichs extension of $\hpm$. By the analyticity of $\overline{T_m}$, $\hat{\mathfrak t}_m$ forms an analytic family of type (B) in the sense of Kato (see Section VII.4 in \cite{kato2013perturbation}). As a consequence, $\mathfrak i_m\overline{\hpm} \mathfrak i_m^{-1}$ and thus  $\overline{\hpm}$ carry the same property because of $\hat{\mathfrak t}_m=\mathfrak i_m\overline{\hpm} \mathfrak i_m^{-1}$, which can be proven as follows. Firstly, the domain $\overline{\dmu(m_o)}$ of $\mathfrak i_m\overline{\hpm} \mathfrak i_m^{-1}$ is the closure of $\dmu(m_o)$ with respect to the graph norm $$\|\psi\|_g:=\sqrt{\|\mathfrak i_m\hpm \mathfrak i_m^{-1}\psi\|^2+\|\psi\|^2}.$$
Secondly, applying the same techniques in the proof of Lemma \ref{lmm:independentofm} in Appendix \ref{app:analyticity}, one can show straight-forwardly that $\|\cdot\|_g$ is equivalent to the norm $\|\cdot\|_{g}^\prime$ defined by
$$
\|\psi\|_{g}^\prime:=\sqrt{\langle\psi|\hat p_b ^4|\psi\rangle+\langle\psi|\psi\rangle}.
$$
Finally, because of
$
\|\psi\|_+\leq \|\psi\|_{g}'$ for all $ \psi\in \dmu(m_o) $, one concludes  $\overline{\dmu(m_o)}\subset Q(m)$, which results in $\hat{\mathfrak t}_m=\mathfrak i_m\overline{\hpm} \mathfrak i_m^{-1}$ according to the uniqueness of the Friedrichs extension. 

Due to the analyticity of  $\overline{\hpm}$, the analyticity of $\omega(m)$ can obtained directly (see e.g. Chapter XII of \cite{reed2003methodsiv} and Theorem VII.1.8 in \cite{kato2013perturbation}), which is summarized precisely as the following theorem.  
\begin{thm}\label{thm:analyticeigen}
Suppose that $\delta_b^{(m)}$ and $\varepsilon_b(m)$ are analytic functions of $m$ at $\mathring{m}$ and $\delta_b^{(\mathring{m})}\neq 0$.
Given an eigenvalue $\mathring{\omega}$ of $\overline{\hp^{(\mathring{m})}}$ which is of an algebraic multiplicity $k$, for each $m$  which is sufficiently close to $\mathring{m}$,
$\overline{\hp^{(m)}}$ has exactly $k$ eigenvalues (counting multiplicity) near $\mathring{\omega}$. These eigenvalues are given by $p\ (\leq k)$ distinct, single-valued and analytic functions $\omega_1(m),\cdots,\omega_p(m)$.
\end{thm}
While this theorem is valid for the general cases that the eigenvalue $\omega$ possesses the algebraic multiplicity $k\geq 1$, it can be seen from the numerical results in the next section that each eigenspace of $\hp^{(m)}$ for all $m$  is exactly one-dimensional.

\subsection{Numerical approach to compute the eigenvalues and eigenstates}
Since $\overline{\hpm}$ has only discrete spectrum, the eigenvector $|\psi\rangle$ associated to each $\omega\in\sigma(\overline{\hpm})$ is  normalizable.  Hence the function $\psi(\mu):=\langle\mu|\psi\rangle$ decreases rapidly for sufficiently large values of $|\mu|$. This fact motivates us to use the finite-dimensional cut-off approximation to collect eigenvalues of $\overline{\hpm}$.  More precisely, we consider a finitely dimensional  subspace, denoted by $\hmu^{(\varepsilon_b,k)}$,  spanned by $|\mu_\eta\rangle=|\varepsilon_b+2\eta\delta_b\rangle\in \hmu^{(\varepsilon_b)}$ with $|\eta|\leq k$ for some large $k$. 
 Let $\hat P^{(k)}$ be the projection operator to $\hmu^{(\varepsilon_b,k)}$ such that
 \begin{equation}
 \hat P^{(k)}|\mu_\eta\rangle=\left\{
 \begin{array}{cc}
 |\mu_\eta\rangle,&|\eta|\leq k,\\
 0,&\text{otherwise}.
 \end{array}
 \right.
 \end{equation}
 Given an eigenvector $|\psi\rangle$ of $\overline{\hpm}$ with respect to an eigenvalue $\omega$, we have
 \begin{equation}\label{eq:phpmh}
 \begin{aligned}
  &(\hat P^{(k)}\overline{\hpm}\hat P^{(k)}-\omega)|\psi\rangle= (\hat P^{(k)}\overline{\hpm}\hat P^{(k)}-\overline{\hpm})|\psi\rangle\\
  =&-(1-\hat P^{(k)})\overline{\hpm}\hat P^{(k)}|\psi\rangle-\hat P^{(k)}\overline{\hpm}(1-\hat P^{(k)})|\psi\rangle-(1-\hat P^{(k)})\overline{\hpm}(1-\hat P^{(k)})|\psi\rangle.
 \end{aligned}
 \end{equation} 
 Since $\psi(\mu):=\langle\mu|\psi\rangle$  rapidly decreases for sufficiently large values of $|\mu|$, the term $1-\hat P^{(k)}$ in the RHS of \eqref{eq:phpmh} indicates that $(\hat P^{(k)}\overline{\hpm}\hat P^{(k)}-\omega)|\psi\rangle$ should be very small for large $k$. Thus it is reasonable to expect that  $\omega$ and $|\psi\rangle$ can be approximated by certain eigenvalue and its corresponding eigenvector of $\hat P^{(k)}\overline{\hpm}\hat P^{(k)}$ respectively for large $k$.  This is the reason for our finite-dimensional  cut-off approximation method. To apply this method, we need to identify those 
 eigenvalues of $\hat P^{(k)}\overline{\hpm}\hat P^{(k)}$ suitable for approximating the eigenvalues of $\hpm$, and check whether all eigenvalues of $\hpm$ can be approximated by this way.  

Since $\hat P^{(k)}\overline{\hpm}\hat P^{(k)}$ is a $(2k+1)$ dimensional symmetric matrix under the basis $|\mu\rangle$, it has $(2k+1)$ eigenvalues. We denote them as $\lambda_i^{(k)}$ with $\lambda_1^{(k)}\leq \lambda_2^{(k)}\leq\cdots\leq \lambda_{2k+1}^{(k)}$, where the subscript $i$ denoted that $\lambda_i^{(k)}$ is the $i$th eigenvalue from the least and the superscript $k$ corresponds to the superscript of $\hat P^{(k)}$.  Obviously, we have $i\leq 2k+1$ in $\lambda_i^{(k)}$.  Thus, once the $i$th eigenvalue of $\hat P^{(k)}\overline{\hpm}\hat P^{(k)}$ for some $k$ is mentioned, $k$ should satisfy $k\geq k_o(i)$ with $k_o(i):=(i-1)/2$.
Moreover, because $\overline{\hpm}$ is semibounded, it has the minimal eigenvalue. Thus its eigenvalues can be denoted by $\omega_i$ with $i=1,2,\cdots$ such that $\omega_1\leq \omega_2\leq \cdots\leq \omega_n\leq\cdots$.  Then by using the Rayleigh-Reitz technique (see, e.g., Theorem XIII.3 in \cite{reed2003methodsiv} and Appendix \ref{app:proof1}), one can obtain 
\begin{equation}\label{eq:orderlambda}
\lambda_i^{(k_o(i))}\geq \lambda_i^{(k_o(i)+1)}\geq \cdots\geq \lambda_i^{(k_o(i)+n)}\geq \cdots\geq\omega_i.
\end{equation}
As a consequence, the limit of $\lambda_i^{(k)}$ as $k\to \infty$ exists, i.e.,
\begin{equation}\label{eq:lambdai}
\lambda_i:=\lim_{k\to\infty}\lambda_i^{(k)}<\infty.
\end{equation}

Consider the sequence $\{\psi_i^{(k_o(i))},\psi_i^{(k_o(i)+1)},\cdots,\psi_i^{(k_o(i)+n)},\cdots\}$, where, for each $k$, $\psi_i^{(k)}$ is an eigenvector of $\hat P^{(k)}\hpm\hat P^{(k)}$ with the eigenvalue $\lambda_i^{(k)}$ and satisfies $\|\psi_i^{(k)}\|=1$. Because of $\|\psi_k^{(k)}\|=1$, the sequence $\{\psi_i^{(k_o(i))},\psi_i^{(k_o(i)+1)},\cdots,\psi_i^{(k_o(i)+n)},\cdots\}$ contains a subsequence by the Banach-Alaoglu theorem, i.e., there exists $\psi_i\in\hmu^{(\varepsilon_b)}$ such that
\begin{equation}\label{eq:weak}
\lim_{l\to \infty}\langle \psi_i^{(n_l)}|\varphi\rangle=\langle \psi_i|\varphi\rangle,\ \forall \varphi\in\hmu^{(\varepsilon_b)}.
\end{equation}
Consider all weakly convergent subsequences of $\{\psi_i^{(k_o(i))},\psi_i^{(k_o(i)+1)},\cdots,\psi_i^{(k_o(i)+n)},\cdots\}$ and collect their limits defined by \eqref{eq:weak}.  Denote the space spanned by these limits as $\Lambda_i$. It turns out that each $\lambda_i$ is an eigenvalue of $\overline{\hpm}$ and all the elements in $\Lambda_i$ are eigenvectors corresponding to the eigenvalue $\lambda_i$ (see Theorem \ref{thm:limiteigen} in Appendix \ref{app:finitecutoff} for proof). Therefore, the $i$th eigenvalue and its corresponding eigenvectors of $\hat P^{(k)}\overline{\hpm}\hat P^{(k)}$ with $k\gg i$ approximate some eigenvalue and eigenvectors of $\overline{\hpm}$ respectively. To check whether all eigenvalues of $\hpm$ can be approximated by the finite-dimensional cut-off approximation,  we can show that
\begin{equation}
\sigma(\overline{\hpm})\cap(\lambda_i,\lambda_{i+1})=\emptyset,
\end{equation}
 if $\lambda_i\neq\lambda_{i+1}$ (see Theorem \ref{thm:inverse} in Appendix \ref{app:finitecutoff} for more details). In other words, each eigenvalue of $\overline{\hpm}$ is a limit point of the sequence $\{\lambda_j^{k_o(j)},\lambda_j^{k_o(j)+1},\cdots,\lambda_j^{k_o(j)+n},\cdots\}$ for some $j$. 

The accuracy of the above approximation can be discussed by the following procedure as in \cite{kato1949upper,harrell1978generalizations}.  Let $|\psi_i^{(k)}\rangle$ be a normalized eigenvector of $\hat P^{(k)}\overline{\hpm}\hat P^{(k)}$ corresponding to the eigenvalue $\lambda_i^{(k)}$. One has
\begin{equation}
\langle \psi_i^{(k)}|\overline{\hpm}|\psi_i^{(k)}\rangle=\lambda_k^{(k)}.
\end{equation} 
Defining
\begin{equation}
\epsilon_i^{(k)}:=\|(\overline{\hpm}-\lambda_i^{(k)})\psi_i^{(k)}\|,
\end{equation}
we will show 
\begin{equation}\label{eq:lambdainaninterval}
\lambda_i\in(\lambda_i^{(k)}-\epsilon_i^{(k)},\lambda_i^{(k)}+\epsilon_i^{(k)}).
\end{equation}  
Let $\{|\omega,\delta\rangle\}$ be an orthonormal basis of the Hilbert space, consisting of eigenvectors of $\overline{\hpm}$, where $\omega$ denotes the eigenvalue and $\delta$ represents other quantum numbers. Given $\lambda_i^{(k)}$, for arbitrary real numbers $\alpha$ and $\beta$ with  $\alpha<\lambda_i^{(k)}<\beta$, we have
\begin{equation}\label{eq:accuracy1}
\langle \psi_i^{(k)}|(\overline{\hpm}-\alpha)(\overline{\hpm}-\beta)|\psi_i^{(k)}\rangle=\sum_\omega\sum_\delta(\omega-\alpha)(\omega-\beta)|\langle \psi_i^{(k)}|\omega,\delta\rangle|^2.
\end{equation}
Then supposing $(\alpha,\beta)\cap\sigma(\overline{\hpm})=\emptyset$, one will get
\begin{equation}\label{eq:inequality}
\langle \psi_i^{(k)}|(\overline{\hpm}-\alpha)(\overline{\hpm}-\beta)|\psi_i^{(k)}\rangle\geq 0.
\end{equation}
Because of $\langle\psi_i^{(k)}|(\hpm-\lambda_i^{(k)})|\psi_i^{(k)}\rangle=0$, $\langle \psi_i^{(k)}|(\overline{\hpm}-\alpha)(\overline{\hpm}-\beta)|\psi_i^{(k)}\rangle$ can be expanded as
\begin{equation}\label{eq:halphahbeta}
\langle \psi_i^{(k)}|(\overline{\hpm}-\alpha)(\overline{\hpm}-\beta)|\psi_i^{(k)}\rangle=(\epsilon_i^{(k)})^2+(\lambda_i^{(k)}-\alpha)(\lambda_i^{(k)}-\beta).
\end{equation}
Substituting \eqref{eq:halphahbeta} into \eqref{eq:inequality}, we get
\begin{equation}\label{eq:betaleq}
\beta\leq \lambda_i^{(k)}+\frac{(\epsilon_i^{(k)})^2}{\lambda_i^{(k)}-\alpha}.
\end{equation}
Note that \eqref{eq:betaleq} holds under the assumption $(\alpha,\beta)\cap\sigma(\overline{\hpm})=\emptyset$. Thus, if $\beta> \lambda_i^{(k)}+(\epsilon_i^{(k)})^2/(\lambda_i^{(k)}-\alpha)$, one will get $(\alpha,\beta)\cap\sigma(\overline{\hpm})\neq\emptyset$, which, together with the fact that $\sigma(\overline{\hpm})$ is closed, ensures
\begin{equation}
\left(\alpha,\lambda_i^{(k)}+\frac{(\epsilon_i^{(k)})^2}{\lambda_i^{(k)}-\alpha}\right]\cap \sigma(\overline{\hpm})\neq \emptyset. 
\end{equation}
Choosing $\alpha$ as $\alpha=\lambda_i^{(k)}-\epsilon_i^{(k)}$, we get
\begin{equation}\label{eq:four24}
\left(\lambda_i^{(k)}-\epsilon_i^{(k)},\lambda_i^{(k)}+\epsilon_i^{(k)}\right]\cap \sigma(\overline{\hpm})\neq \emptyset.
\end{equation}
Noting that \eqref{eq:four24} is true for all $k\geq k_o(i)$ and $\lim_{k\to\infty}\lambda_i^{(k)}\pm\epsilon_i^{(k)}=\lambda_i$, one finally has
 \begin{equation}\label{eq:intevaloc}
 \lambda_i\in (\lambda_i^{(k)}-\epsilon_i^{(k)},\lambda_i^{(k)}+\epsilon_i^{(k)}].
 \end{equation}
 Similar to the derivation of \eqref{eq:betaleq}, one can also obtain  
 \begin{equation}\label{eq:alphageq}
 \alpha\geq \lambda_i^{(k)}+\frac{(\epsilon_i^{(k)})^2}{\lambda_i^{(k)}-\beta},
 \end{equation}
 By choosing $\beta=\lambda_i^{(k)}+\epsilon_i^{(k)}$, \eqref{eq:alphageq} will finally lead to
\begin{equation}\label{eq:intevalco}
 \lambda_i\in [\lambda_i^{(k)}-\epsilon_i^{(k)},\lambda_i^{(k)}+\epsilon_i^{(k)}).
\end{equation}
Then \eqref{eq:lambdainaninterval} is obtained by combining \eqref{eq:intevaloc} and \eqref{eq:intevalco}.

The above discussion can also be applied to estimate the accuracy of approximating $|\psi_i\rangle$ by  $|\psi_i^{(k)}\rangle$. Given an interval $(\alpha_o,\beta_o)\ni\lambda_i$ such that $\lambda_i$ is the only eigenvalue contained in it, we have
\begin{equation}\label{eq:four28}
\begin{aligned}
&\langle \psi_i^{(k)}|(\hpm-\alpha_o)(\hpm-\beta_o)|\psi_i^{(k)}\rangle-(\lambda_i-\alpha_o)(\lambda_i-\beta_o)\sum_{\delta}|\langle\psi_i^{(n)}|\lambda_i,\delta\rangle|^2\\
=&\sum_{\lambda\neq \lambda_i}\sum_\delta(\lambda-\alpha)(\lambda-\beta)\|\langle \psi_k^{(n)}|\lambda,\delta\rangle\|\geq 0.
\end{aligned}
\end{equation}
Substituting \eqref{eq:halphahbeta} into \eqref{eq:four28}, one obtains
\begin{equation}
(\epsilon_i^{(k)})^2+(\lambda_i^{(k)}-\alpha)(\lambda_i^{(k)}-\beta)\geq (\lambda_i-\alpha_o)(\lambda_i-\beta_o)\sum_{\delta}|\langle\psi_i^{(n)}|\lambda_i,\delta\rangle|^2
\end{equation}
which leads to
\begin{equation}\label{eq:eigenstate}
1-\sum_{\delta}\left|\langle\psi_i^{(k)}|\lambda_i,\delta\rangle\right|^2\leq \frac{(\varepsilon_i^{(k)})^2}{(\lambda_i^{(k)}-\alpha_o)(\beta_o-\lambda_i^{(k)})}.
\end{equation} 
Now let us assume  $\lambda_{i-1}^{(k)}\neq \lambda_i^{(k)}\neq \lambda_{i+1}^{(k)}$ without loss of generality. Then one can choose a sufficiently large $k$ such that 
\begin{equation}\label{eq:relationlambdaepsilon}
\lambda_{i-1}^{(k)}+\epsilon_{i-1}^{(k)}<\lambda_i^{(k)}-\epsilon_i^{(k)},\ \lambda_i^{(k)}+\epsilon_i^{(k)}<\lambda_{i+1}^{(k)}-\epsilon_{i+1}^{(k)},
\end{equation}
which, together with \eqref{eq:lambdainaninterval}, implies that the interval $(\lambda_{i-1}^{(k)}+\epsilon_{i-1}^{(k)},\lambda_{i+1}^{(k)}-\epsilon_{i+1}^{(k)})$ contains the single eigenvalue $\lambda_i$. Then, according to \eqref{eq:eigenstate}, we get
\begin{equation}\label{eq:four32}
1-\sum_{\delta}\left|\langle\psi_i^{(k)}|\lambda_i,\delta\rangle\right|^2\leq \frac{(\epsilon_i^{(k)})^2}{(\lambda_i^{(k)}-\lambda_{i-1}^{(k)}-\epsilon_{i-1}^{(k)})(\lambda_{i+1}^{(k)}-\lambda_i^{(k)}-\epsilon_{i+1}^{(k)})}.
\end{equation} 
The left hand side of \eqref{eq:four32} measures the accuracy of approximating $|\psi_i\rangle$ by  $|\psi_i^{(k)}\rangle$ because $|\psi_i^{(k)}\rangle$ is normalized.
 
\subsection{Finite-dimensional cut-off approximation for the numerical computation}
In our model, there are two free parameters  $\varepsilon_b$ and $\b$. By mimicking the derivation of Theorem \ref{thm:analyticfamily} in Appendix \ref{app:analyticity}, one can show that $\overline{\hpm}$, as an operator-valued function of $\varepsilon_b$ and $\b$, form an analytic family of type (B) in the sense of Kato. Therefore, we can set $\varepsilon_b=0=\b$ in our computation, and the results for case with either $\varepsilon_b\neq 0$ or $\b\neq 0$ can be obtained by perturbing that for the case of $\varepsilon_b=0=\b$. 
In other words, the former can be expanded as some power series  of $\varepsilon_b$ and $\b$, and the convergences of the series are ensured by the analyticity of $\hpm$ on $\varepsilon_b$ and $\b$.

By setting $\varepsilon_b=0$, we work in the specific Hilbert space define by \eqref{eq:hbc} with $\varepsilon_b=0$. This Hilbert space is denoted by $\mathcal H^{(0)}$ where a state is given by a wave function $\psi(2\eta\delta_b)\equiv \psi(\eta)$ with $\eta\in\mathbb Z$. Then the action of $\hpm$ reads
\begin{equation}\label{eq:hexpand}
\begin{aligned}
&(\overline{\hpm}\psi)(\eta)=-\frac{1}{4}\gamma^2\Pl^4(\eta+2)(\eta+1)\psi(\eta+2)-i\gamma^2\Pl^2L_0m(\eta+1)\psi(\eta+1)\\
&+\left(\frac{1}{4}\gamma^2\Pl^4(\eta+1)^2+\frac{1}{4}\gamma^2\Pl^4(1+4\delta_b^2\gamma^2)\eta^2\right)\psi(\eta)+i\gamma^2\Pl^2L_0m\eta\psi(\eta-1)-\frac{1}{4}\gamma^2\Pl^4\eta(\eta-1)\psi(\eta-2).
\end{aligned}
\end{equation}
Thus the eigen-equation $(\overline{\hpm}\psi)(\eta)=\omega\psi(\eta)$ results in
\begin{equation}\label{eq:recurrence}
\begin{aligned}
&\psi(\eta+2)=\frac{4}{\gamma^2\Pl^4(\eta+2)(\eta+1)}\Big(-\omega\psi(\eta)-i\gamma^2\Pl^2L_0m(\eta+1)\psi(\eta+1)\\
&+\left(\frac{1}{4}\gamma^2\Pl^4(\eta+1)^2+\frac{1}{4}\gamma^2\Pl^4(1+4\delta_b^2\gamma^2)\eta^2\right)\psi(\eta)+i\gamma^2\Pl^2L_0m\eta\psi(\eta-1)-\frac{1}{4}\gamma^2\Pl^4\eta(\eta-1)\psi(\eta-2)\Big).
\end{aligned}
\end{equation}
For $\eta=0$, \eqref{eq:recurrence} becomes 
\begin{equation}\label{eq:psi201}
\psi(2)=\left(-\frac{2\omega}{\gamma^2\Pl^4}+\frac{1}{2}\right)\psi(0)-i\frac{2L_0m}{\Pl^2}\psi(1),
\end{equation}
which implies that the value of $\psi(2)$ depends only on $\psi(0)$ and $\psi(1)$. Once $\psi(2)$ is obtained by \eqref{eq:psi201}, the values of $\psi(\eta)$ for all $\eta\geq 3$ can be obtained by \eqref{eq:recurrence}. Thus, the values of $\psi(\eta)$ for all $\eta\geq 2$ are determined completely by $\psi(0)$ and $\psi(1)$ by \eqref{eq:recurrence}\footnote{Note that the values of $\psi(0)$ and $\psi(1)$ have to be chosen suitably so that the resulting $\psi(\eta)$ is normalizable. }.  Similarly, one can rewrite \eqref{eq:recurrence} to express $\psi(\eta-2)$ in terms of $\psi(\eta-1)$, $\psi(\eta)$, $\psi(\eta+1)$ and $\psi(\eta+2)$. Then, by setting $\eta=-1$, it is easy to see that the values of $\psi(\eta)$ for all $\eta\leq -3$ are completely determined by $\psi(-1)$ and $\psi(-2)$. 
Thus, for a given eigenvector of $\overline{\hpm}$,  its values for $\eta<0$ decouple from its values for $\eta\geq 0$. Hence, the eigenvectors of $\overline{\hpm}$ can be classified into two supper-selected sectors. The first sector consists of those $\psi$ vanishing for $\eta<0$, while the second sector consists of the ones vanishing for  $\eta\geq 0$. Consider a transformation $T$ defined as
\begin{equation}
(T\psi)(\eta)=\psi(-\eta-1).
\end{equation}
Then $T$ relates the eigenvectors in the two sectors. 
It should be noted that this classification of the eigenvectors is valid only for the case of $\varepsilon_b=0=\b$. However, for the case of $\varepsilon_b\neq 0$ or $\b\neq 0$, the eigenvectors can be divided into two sectors, satisfying 
$
\sum_{\eta\geq 0}|\psi(\eta)|^2\ll 1
$
or 
$
\sum_{\eta<0}|\psi(\eta)|^2\ll 1
$
  respectively.

Given an eigenvector $\psi_-$ in the second sector with the eigenvalue $\omega_-$, a straight-forward calculation gives that 
\begin{equation}\label{eq:hplusepsilon}
(\overline{\hpm} +\hat\epsilon)T\psi_-=\omega_-T\psi_-,
\end{equation}
where $\hat\epsilon$ satisfies
\begin{equation}
(\hat\epsilon\psi)(\eta)=-\gamma^4\Pl^4\delta_b^2(2\eta+1)\psi(\eta). 
\end{equation}
Thus $T\psi_-$ is an eigenvector of $\overline{\hpm}+\hat\epsilon$ with respect to the eigenvalue $\omega_-$.

Define $\hat H(\lambda)=(\overline{\hpm}+\hat\epsilon)+\lambda\hat\epsilon$ on $\overline{\dmu(m)}$. Similar to the discussion on $\overline{\hpm}$, we can show that $\hat H(\lambda)$ is self-adjoint with the domain $\overline{\dmu(m)}$ independent of $\lambda$. Moreover, $\hat H(\lambda)$ forms an analytic family of type (A) in the sense of Kato. This can be verified easily by showing that
$\hat H(\lambda)|\psi\rangle$ for all $|\psi\rangle\in\overline{\dmu(m)}$ is a vector-valued analytic function of $\lambda$ (see e.g. Chapter XII.1 in \cite{reed2003methodsiv}). Therefore, the  Rayleigh-Schr\"odinger perturbation theory can be applied to expand the eigenvalues $\Omega^{(\lambda)}$ and the eigenvectors $\Psi^{(\lambda)}$  of $\hat H(\lambda)$ as the Rayleigh-Schr\"odinger series
\begin{equation}\label{eq:serieseigensys}
\begin{aligned}
\Omega^{(\lambda)}&=\omega_-+\sum_{n=1}^\infty c_n(\lambda\gamma^4\Pl^4\delta_b^2)^n,\\
\Psi^{(\lambda)}&=T\psi_-+\sum_{k=1}^\infty(\lambda\gamma^4\Pl^4\delta_b^2)^n \phi_n,
\end{aligned}
\end{equation}
where $c_n$ and $\phi_n$ denotes the coefficients (see e.g.\cite{reed2003methodsiv}) and the convergence of these series for all $\lambda\in\mathbb R$ is ensured by the analyticity of $\hat H(\lambda)$ on $\lambda$. By setting $\lambda=-1$, one gets the eigenvalue $\Omega^{(-1)}=:\omega_+$ and the corresponding eigenvector $\Psi^{(-1)}=:\psi_+$ of the operator $\hat H(-1)=\hpm$. By substituting the explicit expression of $\phi_n$ into \eqref{eq:serieseigensys}, one can get $\psi_+(\eta)=0$ for all $\eta<0$. Thus $\psi_+$ are in the first sector. Therefore, for each eigenvector of $\hpm$ in the second sector with eigenvalue $\omega_-$, there always exist an adjoint eigenvector $\psi_+$ of $\hpm$  in the first section with eigenvalue $\omega_+$ nearby $\omega_-$. 

By \eqref{eq:serieseigensys}, for  a very small value of $\gamma^4\Pl^4\delta_b^2$, the difference between $\omega_+$ and $\omega_-$ would be very tiny
 so that very high computational cost is needed to separate their values numerically. To overcome this difficulty we consider the Hilbert space $\hmu^{(0+)}\subset\hmu^{(0)}$ defined by
\begin{equation}
\hmu^{(0+)}=\overline{\{\psi\in\hmu^{(0)},\psi(\eta)=0\ \forall \eta<0 \}}
\end{equation}
where the symbol $\overline{\{ \cdot\} }$ represents the completion with respect to the inner product of $\hmu^{(0)}$. 
Then, we diagonalize the operators $\hpm\restriction \hmu^{(0+)} $ and $(\hpm+\hat\epsilon)\restriction \hmu^{(0+)}$, denoting the restrictions of $\hpm$ and $\hpm+\hat\epsilon$ on $\hmu^{(0+)}$ respectively, by the finite-dimensional cut-off approximation method. Let $\psi_+$ be an eigenvector of  $\hpm\restriction \hmu^{(0+)} $ with eigenvalue $\omega_+$, and $\tilde \psi_-$, the eigenvector of $(\hpm+\hat\epsilon)\restriction \hmu^{(0+)}$ with eigenvalue $\omega_-$. Then the vectors $\psi_+$ and $T^{-1}\tilde\psi_-=:\psi_-$ in $\hmu^{(0)}$ are the eigenvectors of $\hpm$ with eigenvalues $\omega_+$ and $\omega_-$ respectively.  Moreover, \eqref{eq:serieseigensys} implies 
\begin{equation}\label{eq:sym}
\psi_+(\eta)\cong (T\psi_-)(\eta). 
\end{equation}

\section{The quantum dynamics}\label{section6}
We now study the dynamics of the model for the cases of $\hat p_\varphi\neq 0$ and $\hat p_\varphi=0$ respectively.  For $\hat p_\varphi\neq 0$, the corresponding classical solution is an extension of the Schwarzschild interior with an extra minimally coupled massless scalar field. This extension is referred to as the Janis-Newman-Winicour (JNW) interior  which differs from the usual JNW spacetime as an extension of Schwarzschild exterior \cite{janis1968reality}. In the classical JNW interior  \cite{zhang2020quantum}, characterized by a parameter $B=2\sqrt{m^2+G p_\varphi^2/4\pi}$, there are two singularities at $r=0$ and $r=B$ respectively, where $r$ is the radial coordinate. Once $p_\varphi$ vanishes, the singularity at $r=B$ will disappear and is replaced by the Schwarzschild horizon, so that the JNW interior becomes the Schwarzschild interior. 

In Section \ref{sec:deparametrization}, gravity is deparametrized by the scalar field which provides a material reference frame of time. 
However, for the case of $\hat p_\varphi=0$, the reference frame of time will disappear. The physical Hamiltonian becomes the Hamiltonian constraint. However, it is still necessary to understand the dynamics of the system as the relational evolution with respect to certain gravitational degree of freedom. The information of the dynamics is encoded in the solutions to $\hat \h\psi=0$. Thus we need to solve this equation.  

\subsection{Dynamics for $\hat p_\varphi\neq 0$}
By \eqref{eq:identity}, the Hilbert space $L^2(\sigma_\c,\dd \mu_\c;\hmu(\cdot))$ of the model consists of functions $\psi:m\mapsto \psi(m)\in\hmu(m)$. 
As shown in Section \ref{sec:betaoperator}, $\hmu(m)$ can be chosen as the one defined by  \eqref{eq:hbc} with $\varepsilon_b=0$ and some $m$-dependent $\delta_b$, denoted by $\delta^{(m)}_b$ whose explicit expression depends on the schemes to quantize $\h$. With this convention, a state $\psi\in L^2(\sigma_\c,\dd \mu_\c;\hmu(\cdot))$ can be represented by a family of functions $$\psi(m,\cdot):\eta\to \psi(m,2\delta_b^{(m)}\eta)$$ with $\eta\in\mathbb Z$. Given a state $\psi$, according to \eqref{eq:phystate} and \eqref{eq:hgr}, an associated dynamical state reads 
\begin{equation}\label{eq:evolutionpsim}
\psi(\varphi,m)=e^{\mp i\frac{\sqrt{4\pi G}}{L_0\gamma\Pl^2}\varphi\sqrt{\hpm} }\hat P_{[0,\infty)}\psi(m).
\end{equation} 
Let $|\omega(m)\rangle\in\hmu(m)$ be the normalized eigenvector of $\hpm$ with respect to the eigenvalue $\omega(m)$.  Then \eqref{eq:evolutionpsim} is simplified as 
\begin{equation}
\psi(\varphi,m)=\sum_{\omega(m)\geq 0}e^{\mp i\frac{\sqrt{4\pi G\omega(m)}}{L_0\gamma\Pl^2}\varphi}\langle\omega(m)|\psi(m)\rangle|\omega(m)\rangle. 
\end{equation}
We choose $\psi(m)$ as
\begin{equation}\label{eq:semiini}
\psi(m,\eta)=e^{-\frac{(m-m_0)^2}{2\sigma_m^2}+i \lambda m}e^{-\frac{(\eta-\eta_0)^2}{2\sigma_\eta^2}-i\beta \eta}
\end{equation}
which carries some semiclassical features. According to \eqref{eq:evolutionpsim}, the initial state evolved by the Hamiltonian is $\hat P_{[0,\infty)}\psi(m)$. Then, it is possible that $\hat P_{[0,\infty)}\psi(m)$ is no longer a semiclassical wave packet even though $\psi(m)$ is. To see how to avoid this possibility, we introduce the expectation value of $\hpm$, 
$$\omega_o(m):=\langle\psi(m)|\hpm|\psi(m)\rangle$$
and its uncertainty 
$$\Delta\omega(m):=\sqrt{\langle\psi(m)|(\hpm)^2|\psi(m)\rangle-\langle\psi(m)|\hpm|\psi(m)\rangle^2}.$$ 
For each $m$, we think of $\psi(m)$ as a wave function of  the eigenvalues $\omega$ of $\hpm$, so it is some wave packet peaked at $\omega_0(m)$ with fluctuation $\Delta\omega(m)$. 
The projection $\hat P_{[0,\infty)}$ cuts off $\psi(m)$ at $m=0$ and vanishes it for all $\omega(m)\leq 0$. Therefore, $\hat P_{[0,\infty)}\psi(m)$ can keep the wave-packet feature of $\psi(m)$ only if $|\omega_o(m)|\gg \Delta\omega(m)$. This condition is the criterion to choose the parameters in \eqref{eq:semiini}.

For a properly chosen $\psi(m)$, its evolution reads
\begin{equation}\label{eq:dyna}
\psi(m,\eta,\varphi)= e^{-\frac{(m-m_0)^2}{2\sigma_m^2}+i \lambda m}\sum_{\omega(m)\geq 0}\langle \mu_\eta^{(m)} |\omega(m)\rangle e^{i\frac{\sqrt{4\pi G\omega(m)}}{L_0\gamma \Pl^2}\varphi}\sum_{\eta'}e^{-\frac{(\eta'-\eta_0)^2}{2\sigma_\eta^2}-i\beta \eta'}\langle \omega(m)|\mu_{\eta'}^{(m)}\rangle.
\end{equation}
To check the consistence between the quantum dynamics and effective dynamics, we calculate the expectation value of the operator $\hat p_b$ as
\begin{equation}\label{eq:expted}
\begin{aligned}
&\frac{\langle \hat p_b\rangle}{\gamma\delta_b \Pl^2}=\mathcal N^2\int_{-\infty}^\infty\dd m e^{-\frac{(m-m_0)^2}{\sigma_m^2}}\sum_\eta\eta\left| \sum_{\omega(m)\geq 0}\langle \mu_\eta^{(m)} |\omega(m)\rangle e^{i\frac{\sqrt{4\pi G\omega(m)}}{\Pl^2\gamma L_0}\varphi} \sum_{\eta'}e^{-\frac{(\eta'-\eta_0)^2}{2\sigma_\eta^2}-i\beta \eta'}\langle \omega(m)|\mu_{\eta'}^{(m)}\rangle\right|^2\\
=:&\frac{1}{\sqrt{\pi}\,\sigma_m}\sum_\eta\eta \int_{-\infty}^\infty \dd m e^{-\frac{(m-m_0)^2}{\sigma_m^2}} \mathcal P(\eta,m,\varphi)
\end{aligned}
\end{equation}
where the normalization factor $\mathcal N$ is given by
\begin{equation}\label{eq:norm}
\mathcal N^{-2}=\int_{-\infty}^\infty \dd m e^{-\frac{(m-m_0)^2}{\sigma_m^2}}\sum_{\omega(m)\geq 0} \left| \sum_{\eta}e^{-\frac{(\eta-\eta_0)^2}{2\sigma_\eta^2}-i\beta \eta}\langle \omega(m)|\mu_\eta^{(m)}\rangle\right|^2=:\int_{-\infty}^\infty \dd m e^{-\frac{(m-m_0)^2}{\sigma_m^2}} \mathfrak{n}(m)
\end{equation}
The integral in \eqref{eq:expted} can be calculated by using the saddle point approximation as
\begin{equation}\label{eq:saddlepoint}
\begin{aligned}
\frac{\langle \hat p_b\rangle}{\gamma\delta_b \Pl^2}=&\sum_{\eta}\eta\,\mathcal P(\eta,m_0,\varphi)\left(1+\frac{\sigma_{m_0}^2}{4}\frac{\mathcal P''(\eta,m_0,\varphi)}{\mathcal P(\eta,m_0,\varphi)}+\cdots\right).
\end{aligned}
\end{equation}
Therefore, as far as the leading order of the evolution is concerned, we only need to compute $\mathcal P(\eta,m_0,\varphi)$. The result of  \eqref{eq:norm} reads
\begin{equation}
\mathcal N^{-2}=\sqrt{\pi}\,\sigma_{m_0}\mathfrak{n}(m_0)\left(1+\frac{\sigma_m^2}{4}\frac{\mathfrak n''(m_0)}{\mathfrak n(m_0)}+\cdots\right).
\end{equation}
Hence one has
\begin{equation}\label{eq:p}
\begin{aligned}
\mathcal P(\eta,m_0,\varphi)\cong \frac{\left| \sum_{\omega(m_0)\geq 0}\langle\mu_\eta^{(m_0)} |\omega(m_0)\rangle e^{i\frac{\sqrt{4\pi G\omega(m_0)}}{\Pl^2\gamma L_0}\varphi} \sum_{\eta'}e^{-\frac{(\eta'-\eta_0)^2}{2\sigma_\eta^2}-i\beta \eta'}\langle \omega(m_0)|\mu_{\eta'}^{(m_0)}\rangle\right|^2}{\sum_{\omega(m_0)\geq 0} \left| \sum_{\eta}e^{-\frac{(\eta-\eta_0)^2}{2\sigma_\eta^2}-i\beta \eta}\langle \omega(m_0)|\mu_\eta^{(m_0)}\rangle\right|^2}.
\end{aligned}
\end{equation}
By this formula, $\mathcal P(\eta,m_0,\varphi)$ can be computed numerically easily. The numerical results of
$\mathcal P(\eta,m_0,\varphi)^{1/2}$ and $\langle \hat p_b\rangle/(\gamma\delta_b\Pl^2)\cong \sum_\eta\eta\mathcal P(\eta,m_0,\varphi)$ are shown in Fig. \ref{fig:nonvcmevlt}, where we choose $\delta_b=\sqrt{\Delta}$ with $\Delta$ being the area gap in LQG.
Moreover, one can also calculate the evolution of $p_b$ with respect to the effective Hamiltonian
\begin{equation}
H=\frac{4\pi}{L_0G\gamma}\sqrt{\frac{2}{\delta_b\delta_c}p_b\sin(\delta_b b)p_c\sin(\delta_c c)+\frac{1}{\delta_b^2}p_b^2\sin(\delta_b b)^2+\gamma^2 p_b^2}.
\end{equation}
As a comparison, the results of the effective dynamics and the classical dynamics are also plotted in Fig. \ref{fig:nonvcmevlt}. 

As shown in Fig. \ref{fig:nonvcmevlt},  both of the classical singularities could be resolved by the effective dynamics where $|\eta|$ is prevented from reaching $0$ by the bounces at the local minimums. Then the classical spacetime is extended periodically.  The quantum evolution matches very well with the effective dynamics for several periods  around $\varphi=0$ when the semiclassical feature of the coherent state is well kept.  The effective evolution and thus the quantum evolution match well with the classical dynamics in the classical regime.
Thus the current quantum theory has a correct semiclassical limit and its semiclassical features can be responded properly by the effective dynamics. However, the coherent 
property of the state
cannot be kept along the whole evolution, since the width of the wave packet grows as the time $\varphi$ runs far away from the initial value $\varphi=0$. This leads to a significant difference between the quantum and effective dynamics in late time. 
\begin{figure}[!t]
\centering
\includegraphics[width=0.6\textwidth]{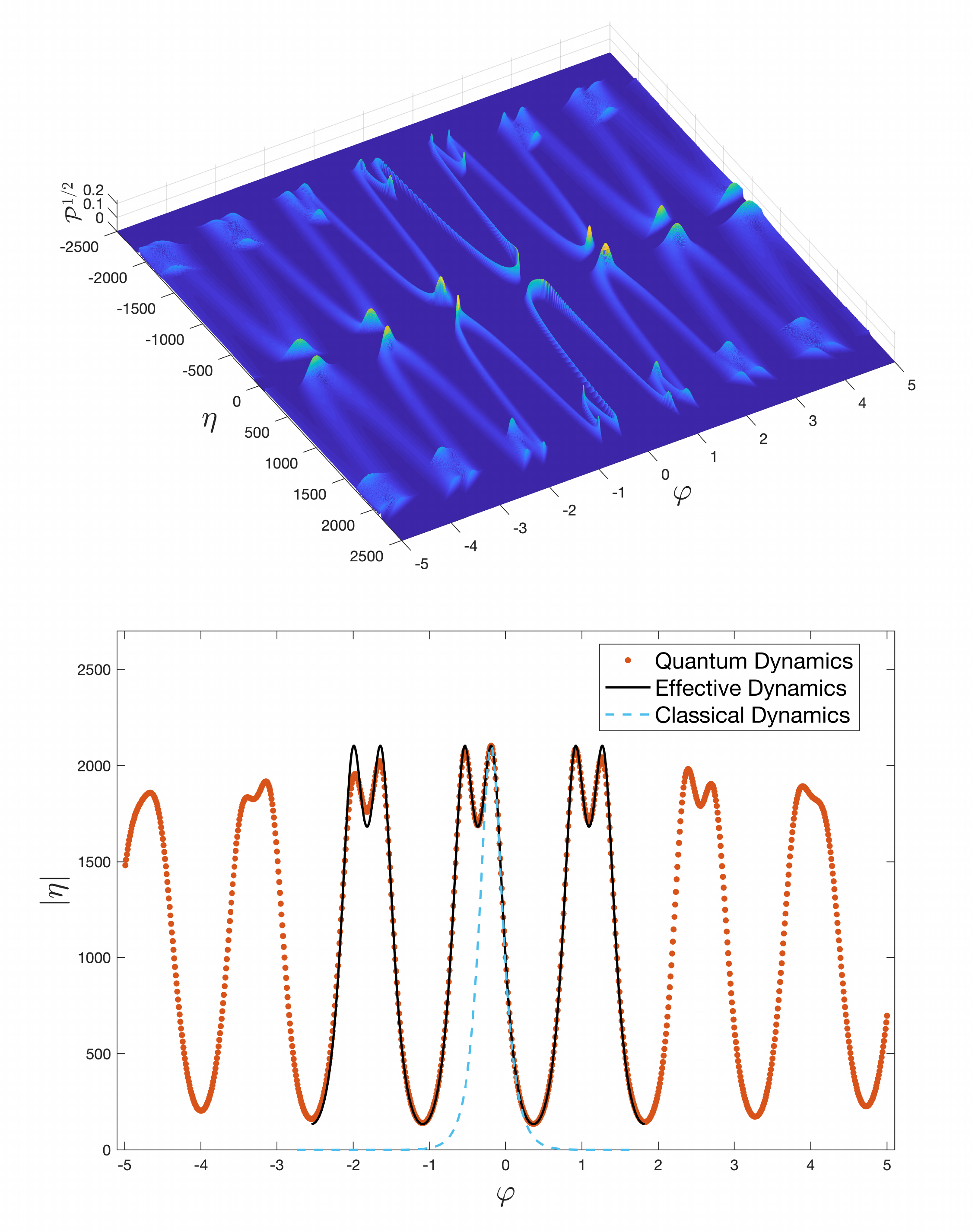}
\caption{Plots of the evolution of the wave packet (top panel) and the evolution of $|\eta|=p_b/(\gamma\delta_b\Pl^2)$ derived by the quantum, the effective dynamics and the classical dynamics (bottom panel). According to the results,   the quantum dynamics as well as the effective dynamics matches well with the classical dynamics in the classical regime. Thus the quantum model admits a correct classical limit. Moreover,
the quantum evolution of $p_b$ matches very well with the effective dynamics in the domain where the wave packet is sharply peaked. However, the width of the wave packet grows as $\varphi$ goes beyond the domain. Then the effective dynamics is no longer valid. }\label{fig:nonvcmevlt}
\end{figure}

\subsection{Dynamics for $\hat p_\varphi=0$}
In the case of $\hat p_\varphi=0$, the dynamics  is encoded in the equation 
\begin{equation}\label{eq:constraint}
\hp\psi=0.
\end{equation}
Alternatively, one could also consider the Hamiltonian constraint operator corresponding to the vacuum Hamiltonian constraint multiplied by volume as a lapse function \cite{zhang2020loop}.
To solve \eqref{eq:constraint},  it is convenient to represent $\psi(m)$ for each $m$ by $\psi(m,\cdot):\omega(m)\mapsto \psi(m,\omega(m))\in\mathbb C$ with $\omega(m)\in\sigma(\overline{\hpm})$. Then the action of $\overline{\hp}$ on $\psi$ reads
\begin{equation}\label{eq:actionhmulti}
(\hp\psi)(m,\omega(m))=\omega(m)\psi(m,\omega(m)). 
\end{equation}
By \eqref{eq:innerH} the inner product of two states $\psi_1$ and $\psi_2$ reads
\begin{equation}\label{eq:innerproduct}
\left(\psi_1,\psi_2\right)=\int_{\sigma_\c}\dd\mu_\c\sum_{\omega(m)\in\sigma(\overline{\hpm})}\psi_1(m,\omega(m))^*\psi_2(m,\omega(m))
\end{equation}
where ${}^*$ means the complex conjugate. 

Given a solution $\psi$ to \eqref{eq:constraint}, Eq. \eqref{eq:actionhmulti} implies 
\begin{equation}
\psi(m,\omega(m))=0,\forall\omega(m)\neq 0.
\end{equation}
Thus the support of the solution $\psi$  is contained in the set 
\begin{equation}
\mathcal S_0=\{(m,\omega(m))|m\in\sigma_\c,\omega(m)\in\sigma(\hpm),\omega(m)=0\},
\end{equation}
which is identified naturally to the set
\begin{equation}
\sigma_\xi:=\{m\in\sigma_\c, 0\in \sigma(\hpm)\}. 
\end{equation}
Given two functions $\psi_i$ with $i=1,2$ on $\sigma_\xi$, \eqref{eq:innerproduct} can be expressed by
\begin{equation}\label{eq:innerproductphy}
(\psi_1,\psi_2)=\int_{\sigma_\xi}\dd\mu_\c\psi_1(m)^*\psi_2(m). 
\end{equation}
Thus, one could naively deem that the functions on $\sigma_\xi$ with the inner product \eqref{eq:innerproductphy} would constitute the physical Hilbert space. However, depending on the explicit expression of $\mu_\c$, it could occur that the right hand side of \eqref{eq:innerproductphy} vanishes for regular functions $\psi_1$ and $\psi_2$. To see how this happens,
let us assume $\sigma_\c=\mathbb R$ at first. 
Because the eigenvalues of $\hpm$ can be expanded at some $m_o$ closed to $m$ by a power series of $m-m_o$ and $\sigma(m_o)$ is discrete, in general they are not $0$. 
Hence it is reasonable to expect that there are only countably many $m\in \sigma_\c$ such that $0\in\sigma(\hpm)$. Then, both $\sigma_\xi$ and $\mathcal S_0$ are countable sets.
This speculation is confirmed by our numerical computation in the $\mu_o$-scheme as we as the scheme with $\delta_c=\sqrt{\Delta}$ and $\delta_b=\sqrt{\Delta}/(2|m|)$, as shown in Fig. \ref{fig:dism}.
For the case of $\sigma_\c\neq \mathbb R$, one has $\sigma_\c\subset\mathbb R$ because $\widehat{p_c\sin(\delta_c c)/\delta_c}$ was assumed to be self-adjoint. Then the resulting $\sigma_\xi$ is just a subset of that for the case of $\sigma_\xi=\mathbb R$. Therefore, $\sigma_\xi$ is always countable. Since $\sigma_\xi$ is countable, it could occur that $\sigma_\xi\subset\sigma_\c$ is of vanishing measure, i.e. $\mu_\c(\sigma_\xi)=0$. For this case, the right hand side of \eqref{eq:innerproductphy} will vanish and thus it cannot define an inner product for the physical Hilbert space. 
Hence we introduce the following procedure to define the inner product in the solution space, which is valid not only for the case of $\mu_\c(\sigma_\xi)=0$ but also for the case of $\mu_\c(\sigma_\xi)\neq 0$. Let $\tilde\delta(m_o,\cdot)$ for each  $m_o\in\sigma_\c$ be the function on $\sigma_\c$ such that (i) $\tilde\delta(m_o,m)=0$ for all $m\neq m_o$ and (ii) $\int_{\sigma_\c}\dd\mu_\c\tilde\delta(m_o,m)=1$. Thus $\tilde\delta(m_o,\cdot)$ is the Dirac-$\delta$ distribution for $\mu_\c(\sigma_\xi)=0$ and proportional to the  Kronecker $\delta$ function otherwise. 
Given a regular function $\psi(m,\omega(m))$, a solution to \eqref{eq:constraint} can be generated as
 \begin{equation}\label{eq:projection}
 \Psi(m)=\sum_{n}\tilde\delta(m_o^{(n)},m) \delta_{0,\omega_m}\psi(m,\omega(m)),
 \end{equation}
 where $m_o^{(n)}\in\sigma_\xi\subset\sigma_\c$.
 By choosing an appropriate dense subspace $\mathcal S\subset\hil$, these $\Psi$ of \eqref{eq:projection} are indeed anti-linear functionals on
  $\mathcal S$ as
  \begin{equation}
  \Psi:\phi\mapsto\Psi[\phi]:= \int_{\sigma_\c}\dd\mu_\c\sum_{\omega_m}\Psi(m)^*\phi(m,\omega_m)=\sum_{m_o^{(n)}}\psi(m_o^{(n)},0)^*\phi(m_o^{(n)},0)
  \end{equation}
  for all $\phi\in\mathcal S$. Hence $\Psi$ is the solution to \eqref{eq:constraint} in the sense that $\Psi[\hp\phi]=0$ for all $\phi\in\mathcal S$. Thus \eqref{eq:projection} defines a rigging map
$\mathbb P:\psi\mapsto\mathbb P\psi$ on $\mathcal{S}$. Therefore, by the refined algebraic quantization procedure \cite{thiemann2008modern}, the physical inner product of two solutions $\Psi_i=\mathbb P\psi_i$ with $i=1,2$ reads
  \begin{equation}
  (\Psi_1|\Psi_2)=\Psi_1[\psi_2]=\sum_{n\in\mathbb Z}\psi_1(m_o^{(n)})^*\psi_2(m_o^{(n)}),
  \end{equation}
  which coincides with \eqref{eq:innerproductphy} if $\mu_\c(\sigma_\xi)\neq 0$. 
Hence the physical Hilbert space of the solutions is given by
\begin{equation}
\hbh:=\overline{\{f:\sigma_\xi\to \mathbb C,\ \sum_n|f(m_o^{(n)})|^2<\infty\}}.
\end{equation} 
The Dirac observable $\hat\xi_{\delta_c}:=\widehat{p_c\sin(\delta_c c)/\delta_c}$ in $\hil$ can be promoted to an operator $\hat \xi'_{\delta_c}$ in $\hbh$ by the dual action such that
\begin{equation}\label{eq:five22}
(\hat \xi_{\delta_c}'f)(m_o^{(n)})=L_0\gamma \delta_c m_o^{(n)}f(m_o^{(n)}),\ \forall\ m_o^{(n)}.
\end{equation}
Eq. \eqref{eq:five22} implies that each $m_o^{(n)}$ is an eigenvalue of $\xi_{\delta_c}'$, and $\xi_{\delta_c}'$ is self-adjoint in $\hbh$ with the spectrum $\overline{\sigma_\xi}$ as the closure of $\sigma_\xi$. 

Given $\psi\in\hmu(m)$, Eq. \eqref{eq:hexpand} indicates
\begin{equation}\label{eq:symmetryminussign}
(\overline{\hp^{(-m)}}\psi)(\eta)^*=(\overline{\hpm}\psi^*)(\eta)
\end{equation}
where $\psi^*\in \hmu(m)$ is defined by $\psi^*(\eta)=\psi(\eta)^*$. For a given $m_o^{(n)}\in \sigma_\xi$, let $\psi_0$ be the eigenvector of $\overline{\hp^{(m_o^{(n)})}}$ with the eigenvalue $0$. Then \eqref{eq:symmetryminussign} ensures that
\begin{equation}
\overline{\hp^{(-m_o^{(n)})}}\psi_0^*=\overline{\hp^{(m_o^{(n)})}}\psi_0=0. 
\end{equation}
Thus, $\psi_0^*$ is an eigenvector of $\hp^{(-m_o^{(n)})}$ with the eigenvalue $0$. Therefore, one has $-m_o^{(n)}\in \overline{\sigma_\xi}$ provided $m_o^{(n)}\in \overline{\sigma_\xi}$. 

For $m=0$, one has the operator $$\hp^{(0)}=\frac{1}{\delta_b^2}\hat\beta_{\delta_b}^2+\gamma^2\hat p_b^2.$$ 
Assume that there is an eigenvector $\phi$ of $\overline{\hp^{(0)}}$ with the eigenvalue $0$. Because of $\hat{\beta}_{\delta_b}^2\geq 0$ and $\hat p_b^2\geq 0$, $\phi$ would satisfy 
\begin{equation}\label{eq:five52}
\hat\beta_{\delta_b}^2\phi=0=\hat p_b^2\phi.
\end{equation}
By the definitions of $\hat\beta_{\delta_b}$ and $\hat p_b^2$, one can easily check that there is no nontrivial $\phi$ satisfying \eqref{eq:five52}. Hence, $0$ is not an eigenvalue of $\overline{\hp^{(0)}}$. 
Taking account of  Theorem \ref{thm:analyticeigen}, one gets the conclusion that $0\notin\sigma(\hpm)$ for sufficiently small $|m|$. Therefore, there exists a gap between the spectrum $\overline{\sigma_\xi}$
and $0$, i.e. $0\notin \overline{\sigma_\xi}$.  Note that $\overline{\sigma_\xi}$ is the spectrum of $\widehat{p_c\sin(\delta_c c)/\delta_c}$ whose classical limit is proportional to the ADM mass of the Schwarzschild BH. Thus, this analysis shows the existence of discrete mass spectrum of the BH. If certain mechanism of BH evaporation could be introduced into our quantum model consistently, the BH would have to evaporate its masses discretely, and the evaporation would have to eventually halt at some  state corresponding to the non-vanishing minimal mass, which is referred to as the BH remnant.

\begin{figure}[!t]
\centering
\includegraphics[width=0.4\textwidth]{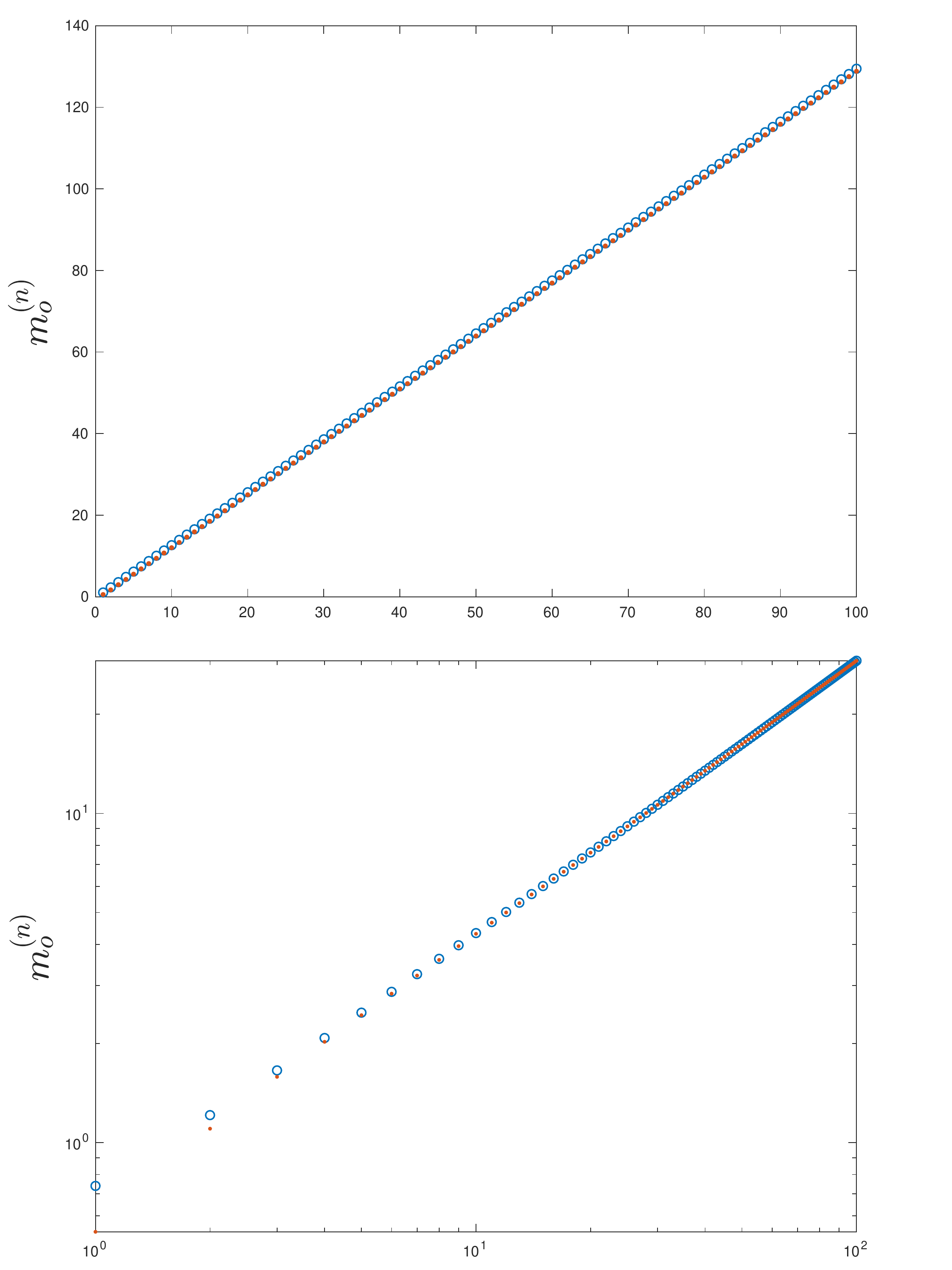}
\caption{Plots of $\sigma_\xi$ for the $\mu_o$ scheme with $\delta_b=\sqrt{\Delta}=\delta_c$ (top panel) and the scheme with $\delta_b=\sqrt{\Delta}/(2\,|m|)$ and $\delta_c=\sqrt{\Delta}$ (bottom panel). As shown in the figure, nearby each value of $m_o^{(n)}$ there exists an adjoint value of $m_o^{(n')}$. $0$ does not belong to $\sigma_\xi$. The parameters are chosen as $\gamma=0.2374$, $\Delta=4\sqrt{3} \pi\gamma\Pl^2$ and $\Pl=1$.}\label{fig:dism} 
\end{figure}

The above analysis is compatible with the numerical results in Fig. \ref{fig:dism}, which shows that the values of $m_o^{(n)}$ are discrete and have the following characters. First, for each $m_o^{(n)}$, there exists an adjoint $m_o^{(n')}$ nearby it. This property comes from the symmetric property \eqref{eq:sym} of the eigenvectors of $\hpm$. Second, the lowest value of $|m_o^{(n)}|$ in $\sigma_\xi$ can be obtained as  $|m_o^{(\rm lwt)}|=0.5499$ for the $\mu_o$ scheme and  $|m_o^{(\rm lwt)}|=0.5362$ for the other scheme, where the parameters are chosen as $\gamma=0.2374$, $\Delta=4\sqrt{3} \pi\gamma \Pl^2$ and $\Pl=1$.

\section{Concluding Remark}\label{sec:conclusion}

The loop quantization of the model of Schwarzschild interior coupled to a massless  scalar field has been studied in the previous sections. By applying the deparametrization procedure, we get the physical Hamiltonian $\h$ of this model with respect to the scalar field. Since $p_c c$ is a Dirac observable in the classical theory, $\h$ is promoted to an operator $\hp$ commutating with the operator $\widehat{p_c\sin(\delta_c c)/\delta_c}$ which corresponds to $p_c c$. Replacing  $\widehat{p_c\sin(\delta_c c)/\delta_c}$ in $\hp$ with its spectrum $L_o\gamma m\in\sigma_\c$, we obtain a family of operators $\hpm$. It is shown that both $\hpm$ and $\hp$ are self-adjoint. The spectrum of $\hpm$ and its analyticity with respect to $m$ are studied. Moreover, we develop a numerical method to diagonalize $\hpm$. Based on these results, the dynamics for the cases of $\hat p_\varphi\neq 0$ and $\hat p_\varphi=0$ are studied respectively. 

For the case of $\hat p_\varphi\neq 0$, the evolution of a wave packet is considered and the results are compared with the effective dynamics governed by the effective Hamiltonian. It turns out that the quantum evolution matches well with the effective dynamics in the domain where the wave packet is sharply peaked. However, the width of the wave packet would increase as the relational time $\phi$ evolves. Thus an inconsistence between the quantum dynamic and the effective dynamics  would occur at late time.  The numerical codes to compute the evolution can be found in \cite{github}. 

For the case of $\hat p_\varphi= 0$, the constraint $\hp=0$ is imposed to get the physical states of the loop quantum Schwarzschild interior model. Its physical Hilbert space $\hbh$ is obtained. The spectrum $\overline{\sigma_\xi}$ of the Dirac observable $\widehat{p_c\sin(\delta_c c)/\delta_c}$ in $\hbh$ is analyzed by both analytical and numerical methods. It turns out that the $\overline{\sigma_\xi}$ is discrete and it does not contain $0$, provided that the parameters $\delta_b^{(m)}$ and $\varepsilon_b(m)$ satisfy the suitable analyticity. Thus, there exists a gap between $\overline{\sigma_\xi}$ and $0$.  Moreover, by the numerical method to diagonalize $\hp$, we can also compute $\overline{\sigma_\xi}$ numerically \cite{github}. Some interesting  results would be obtained if certain mechanism of BH evaporation could be introduced into our quantum model consistently. The only possibility is that the BH evaporates its mass discretely since $\overline{\sigma_\xi}$ is discrete. Moreover, because of $0\notin \overline{\sigma_\xi}$, the evaporation would eventually halt at the state corresponding to the non-vanishing minimal mass. Hence, our result supports the existence of a black hole remnant after Hawking evaporation.

It should be noted that the analyses in the current paper are valid for a quite general class of schemes such that (i) a separable Hilbert subspace $\htau\subset\tilde{\mathcal H}_{\c}$ can be chosen to define an operator $\widehat{p_c\sin(\delta_c c)/\delta_c}$ corresponding to $p_c\sin(\delta_c c)/\delta_c$, which is self-adjoint and commutates with the physical Hamiltonian $\hat \h$; and (ii) 
 the quantum parameter $\delta_b$ is a constant or any function of $p_c\sin(\delta_c c)/\delta_c$.
With the numerical method developed in this paper, it becomes possible to further study the Hawking radiation with the matter backreaction and the distortion of the Hawking spectrum in details. We leave this open issue for our future works. Thus the discrete mass spectrum  predicted by our LQG model provides a solid starting point to study the possibilities of considering the BH remnants as dark matter candidates, as well as solving the puzzle of information loss in BH
evaporation.
\section{acknowledgements}
This work is supported by NSFC with Grants No. 11961131013, No. 11875006 and No. 11775082.  CZ acknowledges the support by the Polish Narodowe Centrum Nauki, Grant No. 2018/30/Q/ST2/00811.
\appendix

\section{Asymptotic behavior of $\hat{\beta}_{\delta_b}$}\label{ap:assyptoticbeta}
By definition, the eigen-equation of $\hat{\beta}_{\delta_b}$ implies that its eigen-functions satisfy
\begin{equation}
\psi(\mu+2\delta_b)=\frac{4i\omega}{\gamma \Pl^2}\frac{1}{\mu+2\delta_b}\psi(\mu)+\frac{\mu}{\mu+2\delta_b}\psi(\mu-2\delta_b),
\end{equation}
which can be rewritten as 
\begin{equation}\label{eq:fd}
\vec\psi(\mu+2\delta_b)=A(\mu)\vec\psi(\mu)
\end{equation}
with
$$
\vec\psi(\mu+2\delta_b):=\left(
\begin{array}{c}
\psi(\mu+2\delta_b)\\
\psi(\mu)
\end{array}
\right) \quad\text{  and   }\quad
A(\mu):=\left(
\begin{array}{cc}
\frac{4i\omega}{\gamma \Pl^2}\frac{1}{\mu+2\delta_b}&\frac{\mu}{\mu+2\delta_b}\\
1&0
\end{array}
\right).
$$
Define
$$
B(\mu)=:\left(
\begin{array}{cc}
\psi_0^+(\mu+2\delta_b)&\psi_0^-(\mu+2\delta_b)\\
\psi_0^+(\mu)&-\psi_0^-(\mu)
\end{array}
\right)
$$
where
$$
\psi_0^\pm(\mu)=\frac{1}{\sqrt{|\mu|}}e^{\pm i k\ln(|\mu|)}.
$$
We obtain that the vector-valued function $
\vec\chi(\mu)=\left(
\begin{array}{c}
\chi_1(\mu)\\
\chi_2(\mu)
\end{array}
\right)
$,
defined by
\begin{equation}\label{eq:chi}
\vec\psi(\mu+2\delta_b)=:B(\mu)\vec\chi(\mu+2\delta_b),
\end{equation}
satisfies 
\begin{equation}\label{eq:A4}
\vec\chi(\mu+2\delta_b)=B(\mu)^{-1}A(\mu)B(\mu-2\delta_b)\vec\chi(\mu):=M(\mu)\vec\chi(\mu).
\end{equation}
which is obtained by applying \eqref{eq:fd}. Our purpose is to derive the condition with which the vector-valued function $\vec\chi(\mu)$ in right-hand-side of \eqref{eq:chi} can be approximated by some constant up to some $O(\mu^{-1})$ term. This can be achieved if the matrix $M(\mu)$ satisfies 
\begin{equation}\label{eq:Mcond}
M(\mu)=M+O(|\mu|^{-2}),
\end{equation}
with some constant $M$.  Substituting \eqref{eq:Mcond} into \eqref{eq:A4}, we obtain
\begin{equation}
k=\frac{\omega}{\gamma\Pl^2\delta_b}
\end{equation}
with
\begin{equation}
M(\mu)=\left(
\begin{array}{cc}
1&0\\
0&-1
\end{array}
\right)+O(|\mu|^{-2}).
\end{equation}
Therefore, one gets
\begin{equation}
\psi(\mu)=\left\{
\begin{aligned}
&\chi_1\psi_0^+(\mu)+\chi_2\psi_0^-(\mu),\ \mu=\varepsilon_b+4 n\delta_n\\
&\chi_1\psi_0^+(\mu)-\chi_2\psi_0^-(\mu),\ \mu=\varepsilon_b+4 n\delta_n+2\delta_b
\end{aligned}
\right.
\end{equation}
with $n\in\mathbb Z$.

\section{The self-adjointness of $\hpm$ and $\hp$}\label{app:selfadjoint}
We now prove that $\hpm$ and $\hp$ are essentially self-adjoint . 

Define an operator $\hat A$ as
\begin{equation}
\hat A |\mu\rangle =\left(\frac{1}{\delta_b^2}\frac{\gamma^2\Pl^4}{16} \big((\mu+2\delta_b)^2+\mu^2\big)+\frac{1}{4}\gamma^4\Pl^4(\mu-2\b)^2\right)|\mu\rangle=:A(\mu)|\mu\rangle,
\end{equation} 
whose domain reads
\begin{equation}
D(\hat A)=\left\{|\psi\rangle,\, \sum_{\mu}\left|A(\mu)\psi(\mu)\right|^2<\infty\right\}.
\end{equation}
Then $\hat A$ is self-adjoint because $A(\mu)$ is real.
Let $\hat B$ be the operator defined on the domain $\dmu$ as
\begin{equation}
\hat B=\frac{2m L_0\gamma}{\delta_b}\hat \beta_{\delta_b}+\frac{1}{\delta_b^2}\left(\hat\beta_{\delta_b}^2-\frac{1}{4}\left(\hat p_b^2+(\hat p_b+\delta_b\gamma\Pl^2)^2\right)\right)
\end{equation}
 Then $\hat B$ can be expressed as
 \begin{equation}
 \hat B=\frac{2mL_0\gamma}{\delta_b}\frac{1}{2i}(\alpha^\dagger-\alpha)-\frac{1}{4\delta_b^2}((\alpha^\dagger)^2+\alpha^2)
 \end{equation}
where
\begin{equation}
\begin{aligned}
\alpha^\dagger\,|\mu\rangle&=\frac{\gamma\Pl^2}{2}(\mu+2\delta_b)\,|\mu+2\delta_b\rangle\\
\alpha\,|\mu\rangle&=\frac{\gamma\Pl^2}{2}\mu\,|\mu+2\delta_b\rangle.
\end{aligned}
\end{equation}
Given $\psi\in\dmu$, we have
\begin{equation}\label{eq:com1}
\|\hat B\psi\|\leq \left\|\left(\frac{2mL_0\gamma}{\delta_b}\frac{1}{2i}\alpha^\dagger-\frac{1}{4\delta_b^2}(\alpha^\dagger)^2\right)\psi\right\|+ \left\|\left(\frac{2mL_0\gamma}{\delta_b}\frac{1}{2i}\alpha+\frac{1}{4\delta_b^2}\alpha^2\right)\psi\right\|.
\end{equation}
Moreover, a straight-forward calculation gives
\begin{equation}\label{eq:com2}
\begin{aligned}
&\left\|\left(\frac{2mL_0\gamma}{\delta_b}\frac{1}{2i}\alpha^\dagger-\frac{1}{4\delta_b^2}(\alpha^\dagger)^2\right)\psi\right\|^2\leq \sum_\mu \Big(\frac{4m^2L_0^2\gamma^4\Pl^4}{16\delta_b^2}\mu^2+\frac{\gamma^4\Pl^8}{256\delta_b^4}\mu^2(\mu-2\delta_b)^2|\\
&\quad\quad +\frac{\gamma^2\Pl^4}{16\delta_b^2}\frac{2mL_0\gamma}{\delta_b}\frac{\gamma\Pl^2}{4}\left(\mu^2(\mu+2\delta_b)+(\mu-2\delta_b)^2\mu\right)\Big)|\psi(\mu)|^2\\
=:&\sum_\mu B^+(\mu)|\psi(\mu)|^2,\\
\end{aligned}
\end{equation}
and 
\begin{equation}\label{eq:com3}
\begin{aligned}
&\left\|\left(\frac{2mL_0\gamma}{\delta_b}\frac{1}{2i}\alpha-\frac{1}{4\delta_b^2}\alpha^2\right)\psi\right\|^2\leq \sum_\mu \Big(\frac{4m^2L_0^2\gamma^4\Pl^4}{16\delta_b^2}(\mu+2\delta_b)^2+\frac{\gamma^4\Pl^8}{256\delta_b^4}(\mu+4\delta_b)^2(\mu+2\delta_b)^2\\
&\quad\quad+\frac{1}{64\delta_b^2}\frac{2mL_0\gamma^4\Pl^6}{\delta_b} \left((\mu+2\delta_b)^2\mu+(\mu+4\delta_b)^2(\mu+2\delta_b)\right)\Big)|\psi(\mu)|^2\\
=:&\sum_\mu B^-(\mu)|\psi(\mu)|^2.
\end{aligned}
\end{equation}
By Eqs. \eqref{eq:com1}, \eqref{eq:com2} and \eqref{eq:com3} we obtain
\begin{equation}
\begin{aligned}
\|\hat B\psi\|\leq \sqrt{\sum_\mu B^+(\mu)|\psi(\mu)|^2}+ \sqrt{\sum_\mu B^-(\mu)|\psi(\mu)|^2}\leq \sqrt{ \sum_\mu 4 B^-(\mu)|\psi(\mu)|^2}
\end{aligned}
\end{equation}
Moreover, one has
\begin{equation}
\begin{aligned}
\|\hat A\psi\|=\sqrt{\sum_\mu A(\mu)^2|\psi(\mu)|^2}.
\end{aligned}
\end{equation}
By the expressions of $B^-(\mu)$ and $A(\mu)$,  there exists some real number $b\geq 0$ such that 
\begin{equation}
4B^-(\mu)\leq \frac{1}{1+2\gamma^4\delta_b^4} A(\mu)^4+b
\end{equation}
which implies
\begin{equation}
\|\hat B\psi\|^2\leq \frac{1}{1+2\gamma^4\delta_b^4}\|\hat A\psi\|^2+b\|\psi\|^2
\end{equation}
Therefore,
$\hpm=\hat B+\hat A$ is self-adjoint on $D(A)$ and thus essentially self-adjoint on $\dmu$ according to the Kato-Rellich Theorem (see Theorem X.12 in \cite{reed2003methodsii}).
Thanks to the self-adjointness of $\hp^{(m)}$, the self-adjointness of $\hp$  can be proven as follows. 
Let $g^\pm$ be elements in the orthogonal complement of the range of $\hp \pm i$ respectively, i.e., $g^\pm\in {\rm Ran}(\hp \pm i)^\perp$. Then we have
\begin{equation}\label{eq:intvanish}
0=\langle (\hat A \pm i)\psi,g^\pm\rangle=\int_{\sigma_\c} \dd \mu_{\c} \left\langle(\hp^{(m)}\pm i)\psi(m),g^\pm(m)\right\rangle,\ \forall \psi\in D(\hp).
\end{equation}
Let $f$ be an square-integrable function $f$ on $\sigma_\c$ with respect to $\mu_\c$, i.e., $f\in L^2(\sigma_\c,\dd\mu_\c)$. Given $\psi\in D(\hp)$, one has 
\begin{equation}
\left|\int_{\sigma_\c}\dd\mu_\c\left\langle f(m)\hp^{(m)}\psi(m),f(m)\hp^{(m)}\psi(m)\right\rangle\right|^2\leq\int_{\sigma_\c}\dd\mu_c|f(m)|^2\int_{\sigma_\c}\|\hp^{(m)}\psi(m)\|^2<\infty 
\end{equation}
which implies that $f\psi:m\mapsto f(m)\psi(m)$ is also in $D(\hp)$. Thus substituting $f\psi$ into \eqref{eq:intvanish}, we conclude
\begin{equation}
\int_{\sigma_\c} \dd \mu_{\c} \overline{f(m)}\left\langle(\hp^{(m)}\pm i)\psi(m),g^\pm(m)\right\rangle=0,\ \forall f\in L^2(\sigma_\c,\dd\mu_\c).
\end{equation}
Therefore, it holds almost everywhere for $m$ that
\begin{equation}\label{eq:conclusionweneed}
\langle (\hp^{(m)}\pm i)\psi(m),g^\pm(m)\rangle=0.
\end{equation}
Since this conclusion is true for all $\psi(m)\in \dmu(m)$, one has that $g^\pm(m)\in  {\rm Ran}(\hp^{(m)}\pm i)^\perp$ almost everywhere for $m$. However, $\hp^{(m)}$ is essentially self-adjoint for all $m$. Hence one gets
\begin{equation}
{\rm Ran}(\hp^{(m)}\pm i)^\perp=\{0\}, \forall m.
\end{equation}
This ensures that $g^\pm (m)=0$ for all $m$. The self-adjointness of $\hat A$ is thus obtained because of the basic criterion for self-adjointness (see, e.g., Theorem VIII.3 in \cite{reed2003methods}).

\section{The analiticity of $\overline{T_m}$}\label{app:analyticity}

\begin{thm}\label{thm:analyticfamily}
$\overline{T_m}$ is an analytic family of forms of type (a) in the sense of Kato. 
\end{thm}

By definition, this theorem can be obtained directly from the following lemmas. 
\begin{lmm}\label{lmm:independentofm}
The norm $\|\cdot\|_m$ defined on $\dmu(m_o)$ for each $m$ is equivalent to the norm 
\begin{equation}\label{eq:normnew}
\|\psi\|_+=\sqrt{\langle \psi|\hat p_b^2|\psi\rangle+\langle\psi|\psi\rangle},
\end{equation}
i.e., there exist constants $c,C>0$ such that
\begin{equation}\label{eq:equivalentnorm}
c\|\psi\|_+\leq \|\psi\|_m\leq C\|\psi\|_+, \forall \psi\in \dmu(m_o).
\end{equation}
\end{lmm}
\begin{proof}
Given $|\psi\rangle=\sum_\mu\psi_\mu|\mu\rangle\in \dmu(m_o)$, by \eqref{eq:hpmg} one gets 
\begin{equation}\label{eq:lowbonded}
\begin{aligned}
T_m(\psi,\psi)+(L_0^2m^2\gamma^2+1)\langle\psi|\psi\rangle
&\geq \gamma^2\langle\psi|(\hat p_b-\gamma\Pl^2\b)^2|\psi\rangle+\langle\psi|\psi\rangle\geq \gamma^2 c\langle\psi|\hat p_b^2|\psi\rangle+\langle\psi|\psi\rangle
\end{aligned}
\end{equation}
where
\begin{equation}
c=1 + \frac{(\gamma\Pl^2\b)^2}{2} - \frac{\gamma\Pl^2\b \sqrt{4 + (\gamma\Pl^2\b)^2}}{2}>0.
\end{equation}
Moreover, consider the operator $(1+\hat p_b^2)^{-\frac12}\,(\mathfrak i_{m}\hp^{(m)}\mathfrak i_{m}^{-1}+L_0^2\gamma^2m^2+1)(1+\hat p_b^2)^{-\frac12}$ on $\dmu(m_o)$. It can be verified straight-forwardly that 
\begin{equation}
\left|\langle \mu| (1+\hat p_b^2)^{-\frac12}\,(\mathfrak i_{m}\hp^{(m)}\mathfrak i_{m}^{-1}+L_0^2\gamma^2m^2+1)(1+\hat p_b^2)^{-\frac12}\,|\mu'\rangle\right|<\infty.
\end{equation} 
Thus the operator $(1+\hat p_b^2)^{-\frac12}\,(\mathfrak i_{m}\hp^{(m)}\mathfrak i_{m}^{-1}+L_0^2\gamma^2m^2+1)(1+\hat p_b^2)^{-\frac12}$ is bounded, i.e., there exist $C>0$ such that 
\begin{equation}\label{eq:boundedmap}
\left|\langle \psi| (1+\hat p_b^2)^{-\frac12}\,(\mathfrak i_{m}\hp^{(m)}\mathfrak i_{m}^{-1}+L_0^2\gamma^2m^2+1)(1+\hat p_b^2)^{-\frac12}\,|\psi\rangle\right|\leq C\langle\psi|\psi\rangle,\forall \psi\in \dmu(m_o). 
\end{equation}
By definition, $(1+\hat p_b^2)^{-\frac12}:\dmu(m_o)\to\dmu(m_o)$ is surjective. Hence \eqref{eq:boundedmap} implies
\begin{equation}
\left|\langle \psi|(\mathfrak i_{m}\hp^{(m)}\mathfrak i_{m}^{-1}+L_0^2\gamma^2m^2+1)|\psi\rangle\right|\leq C\langle\psi|(1+\hat p_b^2)|\psi\rangle,\forall \psi\in \dmu(m_o). 
\end{equation}
Thus
\begin{equation}\label{eq:upbonded}
T_m(\psi,\psi)+(L_0^2m^2\gamma^2+1)\langle\psi|\psi\rangle\leq C\langle\psi|(1+\hat p_b^2)|\psi\rangle,\forall \psi\in \dmu(m_o).
\end{equation}
Then \eqref{eq:equivalentnorm} is proven by \eqref{eq:lowbonded} and \eqref{eq:upbonded}.
\end{proof}
By Lemma \ref{lmm:independentofm}, for each $m$, $Q(m)$ is indeed the closure of $\dmu(m_o)$ with respect to the norm $\|\cdot\|_+$ given in \eqref{eq:normnew}. Thus one gets $Q(m)=Q(m_o)$ for all $m$. 

\begin{lmm}\label{lmm:analyticity}
Suppose $\delta_b^{(m)}$ and $\varepsilon_b(m)$ to be analytic functions of $m$ at $m_o$ and $\delta_b^{(m_o)}\neq 0$. Then $\overline{T_m}(\psi,\psi)$,
 for each $\psi\in Q(m_o)$,  is an analytic function of $m$ at $m_o$. 
\end{lmm}
\begin{proof}
The action of $\mathfrak i_{m}\hp^{(m)}\mathfrak i_{m}^{-1}$ on $|\mu_n\rangle$ with $\mu_n=\varepsilon_b^o+2n\delta_b^o$ reads
\begin{equation}\label{eq:ihiaction}
\begin{aligned}
\mathfrak i_{m}\hp^{(m)}\mathfrak i_{m}^{-1}| \mu_n\rangle=&-\frac{1}{4(\delta_b^o)^2}(\f(m)+\hat p_b)(\f(m)+\hat p_b-\gamma\Pl^2\delta_b^o)|\mu_{n+2}\rangle+\frac{2 L_0m\gamma}{2\delta_b^oi}(\f(m)+\hat p_b)|\mu_{n+1}\rangle\\
&+\left(\frac{1}{4(\delta_b^o)^2)}((\f(m)+\hat p_b+\gamma\Pl^2\delta_b^o)^2+((\f(m)+\hat p_b)^2)+\gamma^2(\delta_b^{(m)})^2((\f(m)+\hat p_b-\gamma\Pl^2\b/\delta_b^{(m)})^2\right)|\mu_n\rangle\\
&-\frac{2 L_0 m\gamma}{2\delta_b^oi}(\f(m)+\hat p_b+\gamma\Pl^2\delta_b^o)|\mu_{n-1}\rangle-\frac{1}{4(\delta_b^o)^2}(\f(m)+\hat p_b+2\gamma\Pl^2\delta_b^o)(\f(m)+\hat p_b+\gamma\Pl^2\delta_b^o)|\mu_{n-2}\rangle
\end{aligned}
\end{equation}
where $\f(m):=\varepsilon_b(m)\delta_b^o/\delta_b^{(m)}-\varepsilon_b^o$ is analytic at $m_o$. Given $|\psi\rangle\in Q(m_o)$,  Eq. \eqref{eq:ihiaction} implies that $\langle \psi|\mathfrak i_{m}\hp^{(m)}\mathfrak i_{m}^{-1}|\psi\rangle$ is a finite linear combination of $\langle \psi|\hat p_b^2|\psi\rangle$, $\langle \psi|\hat p_b|\psi\rangle$ and $\langle \psi|\psi\rangle$ with coefficients depending on $m$ analytically. Since $Q(m)$ is the closure of $\dmu(m_o)$ with respect to the norm $\|\cdot\|_+$, both $\langle \psi|\hat p_b^2|\psi\rangle$ and $\langle \psi|\hat p_b|\psi\rangle$ are well-defined for all $\psi\in Q(m)$. Then the analyticity of $\overline{T_m}(\psi,\psi)$ is proven.  
\end{proof}

\section{Proof of equations \eqref{eq:orderlambda} and \eqref{eq:lambdai}}\label{app:proof1}
We notice that
\begin{equation}
\begin{aligned}
\lambda_i^{(k)}=&\sup_{\phi_1,\phi_2,\cdots,\phi_{i-1}\in\hmu^{(\varepsilon_b,k)}}\inf_{\substack{\psi\in\hmu^{(\varepsilon_b,k)};\|\psi\|=1\\ \langle\psi,\phi_n\rangle=0,\forall n=1,2,\cdots,i-1}}\langle\psi,\hat P^{(k)}\overline{\hpm}\hat P^{(k)}\psi\rangle\\
=&\sup_{\phi_1,\phi_2,\cdots,\phi_{i-1}\in\hmu^{(\varepsilon_b,k)}}\inf_{\substack{\psi\in\hmu^{(\varepsilon_b,k)};\|\psi\|=1\\ \langle\psi,\phi_n\rangle=0,\forall n=1,2,\cdots,i-1}}\langle\psi,\hpm\psi\rangle
\end{aligned}
\end{equation}
because of $\hat P^{(k)}\overline{\hpm}\hat P^{(k)}=\hat P^{(k)}\hpm\hat P^{(k)}$, and $\hat P^{(k)}\psi=\psi$ for all $\psi\in\hmu^{(\varepsilon_b,k)}$. 

Given $k'>k$, one has $\hmu^{(\varepsilon_b,k)}\subset \hmu^{(\varepsilon_b,k')}$. Let $\tilde\phi_1,\ \tilde\phi_2,\ \cdots \tilde\phi_{i-1}\in\hmu^{(\varepsilon_b,k')}$ be some vectors such that $\tilde\phi_{i-n},\tilde\phi_{i-n+1},\cdots,\tilde\phi_{i-1}\notin \hmu^{(\varepsilon_b,k)}$. Then we have
\begin{equation}
\begin{aligned}
\inf_{\substack{\psi\in\hmu^{(\varepsilon_b,k)};\|\psi\|=1\\ \langle\psi,\tilde \phi_l\rangle=0,\forall l=1,2,\cdots,i-1}}\langle\psi,\hpm\psi\rangle=\inf_{\substack{\psi\in\hmu^{(\varepsilon_b,k)};\|\psi\|=1\\ \langle\psi,\tilde \phi_l\rangle=0,\forall l=1,2,\cdots,i-n-1}}\langle\psi,\hpm\psi\rangle\leq \lambda_{i-n}^{(k)}\leq \lambda_i^{(k)}.
\end{aligned}
\end{equation}
Therefore, one obtains
\begin{equation}
\begin{aligned}
\lambda_i^{(k)}=&\sup_{\phi_1,\phi_2,\cdots,\phi_{i-1}\in\hmu^{(\varepsilon_b,k)}}\inf_{\substack{\psi\in\hmu^{(\varepsilon_b,k)};\|\psi\|=1\\ \langle\psi,\phi_n\rangle=0,\forall n=1,2,\cdots,i-1}}\langle\psi,\hpm\psi\rangle\\
=&\sup_{\phi_1,\phi_2,\cdots,\phi_{i-1}\in\hmu^{(\varepsilon_b,k')}}\inf_{\substack{\psi\in\hmu^{(\varepsilon_b,k)};\|\psi\|=1\\ \langle\psi,\phi_n\rangle=0,\forall n=1,2,\cdots,i-1}}\langle\psi,\hpm\psi\rangle.
\end{aligned}
\end{equation}
Furthermore, for given $\phi_1,\phi_2,\cdots,\phi_{i-1}\in\hmu^{(\varepsilon_b,k')}$, we have
\begin{equation}
\begin{aligned}
&\{\psi\in\hmu^{(\varepsilon_b,k)}\ \big|\ \|\psi\|=1;\langle\psi,\phi_n\rangle=0,\forall n=1,2,\cdots,i-1\}\\
\subset&\{ \psi\in\hmu^{(\varepsilon_b,k')} \ \big|\  \|\psi\|=1;\langle\psi,\phi_n\rangle=0,\forall n=1,2,\cdots,i-1\}.
\end{aligned}
\end{equation}
Thus, one has
\begin{equation}
\begin{aligned}
&\inf_{\substack{\psi\in\hmu^{(\varepsilon_b,k)};\|\psi\|=1\\ \langle\psi,\phi_n\rangle=0,\forall n=1,2,\cdots,i-1}}\langle\psi,\hpm\psi\rangle\geq \inf_{\substack{\psi\in\hmu^{(\varepsilon_b,k')};\|\psi\|=1\\ \langle\psi,\phi_n\rangle=0,\forall n=1,2,\cdots,i-1}}\langle\psi,\hpm\psi\rangle,
\end{aligned}
\end{equation}
which implies
\begin{equation}
\begin{aligned}
\lambda_i^{(k)}=&\sup_{\phi_1,\phi_2,\cdots,\phi_{i-1}\in\hmu^{(\varepsilon_b,k')}}\inf_{\substack{\psi\in\hmu^{(\varepsilon_b,k)};\|\psi\|=1\\ \langle\psi,\phi_n\rangle=0,\forall n=1,2,\cdots,i-1}}\langle\psi,\hpm\psi\rangle\\
\geq&\sup_{\phi_1,\phi_2,\cdots,\phi_{i-1}\in\hmu^{(\varepsilon_b,k')}}\inf_{\substack{\psi\in\hmu^{(\varepsilon_b,k')};\|\psi\|=1\\ \langle\psi,\phi_n\rangle=0,\forall n=1,2,\cdots,i-1}}\langle\psi,\hpm\psi\rangle=\lambda_i^{(k')}.
\end{aligned}
\end{equation}
Similarly, we have
\begin{equation}
\begin{aligned}
\lambda_i^{(k)}=&\sup_{\phi_1,\phi_2,\cdots,\phi_{i-1}\in\hmu^{(\varepsilon_b,k)}}\inf_{\substack{\psi\in\hmu^{(\varepsilon_b,k)};\|\psi\|=1\\ \langle\psi,\phi_n\rangle=0,\forall n=1,2,\cdots,i-1}}\langle\psi,\hpm\psi\rangle\\
=&\sup_{\phi_1,\phi_2,\cdots,\phi_{i-1}\in\hmu^{(\varepsilon_b)}}\inf_{\substack{\psi\in\hmu^{(\varepsilon_b,k)};\|\psi\|=1\\ \langle\psi,\phi_n\rangle=0,\forall n=1,2,\cdots,i-1}}\langle\psi,\hpm\psi\rangle\\
\geq&\sup_{\phi_1,\phi_2,\cdots,\phi_{i-1}\in\hmu^{(\varepsilon_b)}}\inf_{\substack{\psi\in \overline{\dmu(m)};\|\psi\|=1\\ \langle\psi,\phi_n\rangle=0,\forall n=1,2,\cdots,i-1}}\langle\psi,\hpm\psi\rangle=\omega_i.
\end{aligned}
\end{equation} 
Thus the proof of \eqref{eq:orderlambda}  is completed, and the existence of the limit $\lim_{k\to\infty}\lambda_i^{(k)}$ can be obtained directly from \eqref{eq:orderlambda}. 

\section{The finite cut-off approximation}\label{app:finitecutoff}
\begin{thm}\label{thm:limiteigen}
Each $\lambda_i$ given in \eqref{eq:lambdai} is an eigenvalue of $\overline{\hpm}$. The space $\Lambda_i$  is an eigenspace corresponding to the eigenvalue $\lambda_i$. 
\end{thm}
\begin{proof}
By definition, one has
\begin{equation}
\|\hat P^{(k)}(\hpm -\lambda_i)\psi_i^{(k)}\|=\|(\lambda_i^{(k)}-\lambda_i)\psi_i^{(k)}\|\leq |\lambda_i^{(k)}-\lambda_i|,
\end{equation}
which implies
\begin{equation}
\lim_{k\to \infty}\|\hat P^{(k)}(\hpm -\lambda_i)\psi_i^{(k)}\|=0.
\end{equation}
Given an arbitrary $\varphi\in \dmu(m)$, by the definition \eqref{eq:Dbm} of $\dmu(m)$, there exists an integer $N_\varphi$ such that $\varphi\in \hmu^{(\varepsilon_b,k)},\ \forall k \geq N_\varphi$. Thus, for each $\varphi\in \dmu(m)$, one gets
\begin{equation}
\langle\varphi| \hat P^{(k)}(\hpm -\lambda_i)|\psi_i^{(k)}\rangle=\langle\varphi|(\hpm -\lambda_i)|\psi_i^{(k)}\rangle,\forall k\geq N_\varphi
\end{equation}
where we used $\hat P^{(k)}\varphi=\varphi$ for all $k\geq N_\varphi$ and that $\hpm$ is symmetric.  
Hence, we have
\begin{equation}\label{eq:limktoinfty}
0=\lim_{k\to \infty}\langle\varphi| \hat P^{(k)}(\hpm -\lambda_i)|\psi_i^{(k)}\rangle=\langle \varphi|(\hpm -\lambda_i)|\psi_i\rangle,\ \forall \varphi\in \dmu(m).
\end{equation}
By this equation, $\psi_i$ is in the domain of the adjoint of $\hpm$. Since $\hpm$ is essentially self-adjoint, its adjoint is equal to its closure $\overline{\hpm}$. Thus $\psi_i$ is in $\overline{\dmu(m)}$. Moreover, \eqref{eq:limktoinfty} also implies
\begin{equation}\label{eq:psikeigen}
(\overline{\hpm}-\lambda_i)\psi_i=0,
\end{equation}
which ensures either $\psi_i=0$ or that $\psi_i$ is an eigenvector of $\overline{\hpm}$ with the eigenvalue $\lambda_i$. We now show that $\psi_i\neq 0$ and hence  $\psi_i$ can only be an eigenvector of $\overline{\hpm}$. 
Given an eigenvalue $\omega$ of $\overline{\hpm}$, let $|\omega,\alpha\rangle$ be an orthonormal basis of the eigenspace with respect to $\omega$. Define a projection $\hat P$ as
\begin{equation}
\hat P|\psi\rangle:=\sum_{\omega\in(-\infty,\lambda_i]}\sum_\alpha|\omega,\alpha\rangle\langle \omega,\alpha|\psi\rangle .
\end{equation}
Because $\overline{\hpm}$ is bounded from below and each eigenvalue $\omega$ is of finite multiplicity, 
the summation in the RHS consists of only finite terms. Then for each vectors $\psi_i^{(n_l)}$ in \eqref{eq:weak}, we have
\begin{equation}\label{eq:E7}
\begin{aligned}
&\left\|(\overline{\hpm}-\lambda_i)\hat P\psi_i^{(n_l)}-\sum_{\omega\in(-\infty,\lambda_k]}\sum_\alpha|\omega,\alpha\rangle\langle \omega,\alpha|\psi_i\rangle(\omega-\lambda_i) \right \|\\
\leq&\sum_{\omega\in(-\infty,\lambda_i]}\sum_\alpha\left|\left(\langle \omega,\alpha|\psi_i^{(n_l)}\rangle-\langle\omega,\alpha|\psi_i\rangle\right)\right|\left|\omega-\lambda_i \right |.
\end{aligned}
\end{equation}
Taking account of \eqref{eq:E7},  \eqref{eq:weak} and the fact that the summation contains finitely many terms, we obtain
\begin{equation}
\lim_{l\to\infty}\left\|(\overline{\hpm}-\lambda_i)\hat P\psi_i^{(n_l)}-\sum_{\omega\in(-\infty,\lambda_i]}\sum_\alpha|\omega,\alpha\rangle\langle \omega,\alpha|\psi_i\rangle(\omega-\lambda_i) \right \|=0,
\end{equation}
i.e.
\begin{equation}
\lim_{l\to \infty}(\overline{\hpm}-\lambda_i)\hat P\psi_i^{(n_l)}=\sum_{\omega\in(-\infty,\lambda_i]}\sum_\alpha|\omega,\alpha\rangle\langle \omega,\alpha|\psi_i\rangle(\omega-\lambda_i) .
\end{equation}
Moreover, because of \eqref{eq:psikeigen} and
\begin{equation}
\begin{aligned}
\sum_{\omega\in(-\infty,\lambda_i]}\sum_\alpha|\omega,\alpha\rangle\langle \omega,\alpha|\psi_i\rangle(\omega-\lambda_i)
=\sum_{\omega\in(-\infty,\lambda_i]}\sum_\alpha|\omega,\alpha\rangle\langle \omega,\alpha|(\overline{\hpm}-\lambda_i)\psi_i\rangle , 
\end{aligned}
\end{equation}
we finally have
\begin{equation}
\lim_{l\to \infty}(\overline{\hpm}-\lambda_i)\hat P\psi_i^{(n_l)}=0.
\end{equation}
Furthermore, by using
\begin{equation}
\|\hat P^{(n_l)}(\overline{\hpm}-\lambda_i)\psi_i^{(n_l)}\|=\|(\lambda_i^{(n_l)}-\lambda_i)\psi_i^{(n_l)}\|\leq |\lambda_i^{(n_l)}-\lambda_i|,
\end{equation}
one has 
\begin{equation}\label{eq:E13}
\lim_{l\to\infty}\|\hat P^{(n_l)}(\overline{\hpm}-\lambda_i)\psi_i^{(n_l)}\|=0.
\end{equation}
Combining \eqref{eq:E13} with the inequality  
\begin{equation}
\begin{aligned}
&\||\hat P^{(n_l)}(\overline{\hpm}-\lambda_i)(1-\hat P)\psi_i^{(n_l)}\|\leq  \|\hat P^{(n_l)}(\overline{\hpm}-\lambda_i)\psi_i^{(n_l)}\|+\|(\overline{\hpm}-\lambda_i)\hat P\psi_i^{(n_l)}\|,
\end{aligned}
\end{equation}
we finally obtain
\begin{equation}
\lim_{l\to\infty}\|\hat P^{(n_l)}(\overline{\hpm}-\lambda_i)(1-\hat P)\psi_i^{(n_l)}\|\leq\lim_{l\to \infty}\|(\overline{\hpm}-\lambda_i)\hat P\psi_i^{(n_l)}\|=0,
\end{equation}
which implies
\begin{equation}\label{eq:someeq}
\lim_{l\to\infty}\hat P^{(n_l)}(\overline{\hpm}-\lambda_i)(1-\hat P)\psi_i^{(n_l)}=0.
\end{equation}
Defining $\tilde\omega:=\inf\{\omega\in\sigma(\hpm),\omega>\lambda_i\} $, we have 
\begin{equation}\label{eq:lastequation0}
\begin{aligned}
&\langle \hat P^{(n_l)}(\overline{\hpm}-\lambda_i)(1-\hat P)\psi_i^{(n_l)}|\psi_i^{(n_l)}\rangle=\langle(1-\hat P)(\overline{\hpm}-\lambda_i)(1-\hat P)\psi_i^{(n_l)}|\psi_i^{(n_l)}\rangle\\
\geq &(\tilde\omega-\lambda_i)\|(1-\hat P)\psi_i^{(n_l)}\|,
\end{aligned}
\end{equation}
where the last inequality is resulted from
$
\langle (\overline{\hpm}-\lambda_k)(1-\hat P)\varphi,\varphi\rangle\geq (\tilde\omega-\lambda_i) \langle (1-\hat P)\varphi,\varphi\rangle
$ for all $\varphi\in \overline{\dmu(m)}$.
The combination of \eqref{eq:lastequation0} and \eqref{eq:someeq} leads to
\begin{equation}
\lim_{l\to\infty}\| (1-\hat P)\psi_i^{(n_l)}\|=0.
\end{equation}
Furthermore, because of 
\begin{equation}
\|\psi_i^{(n_l)}-\psi_i\|\leq \|\hat P\psi_i^{(n_l)}-\hat P\psi_i\|+\|(1-\hat P)\psi_i^{(n_l)}\|+\|(1-\hat P)\psi_i\|
\end{equation}
we obtain
\begin{equation}\label{eq:E20}
\lim_{l\to\infty}\|\psi_i^{(n_l)}-\psi_i\|\leq \lim_{l\to\infty}\|\hat P\psi_i^{(n_l)}-\hat P\psi_i\|+\|(1-\hat P)\psi_i\|.
\end{equation}
The first term in the RHS of \eqref{eq:E20} satisfies 
\begin{equation}\label{eq:E21}
\begin{aligned}
\left\|\hat P\psi_i^{(n_l)}-\hat P\psi_i\right \|\leq\sum_{\omega\in(-\infty,\lambda_i]}\sum_\alpha\left|\left(\langle \omega,\alpha|\psi_k^{(n_l)}\rangle-\langle \omega,\alpha|\psi_k\rangle\right)\right|.
\end{aligned}
\end{equation}
Taking account of \eqref{eq:weak} and the fact that the summation in the RHS contains finite terms, Eq. \eqref{eq:E21} implies
\begin{equation}
\lim_{l\to \infty} \left\|\hat P\psi_i^{(n_l)}-\hat P\psi_i\right \|=0,
\end{equation}
Therefore, we have
\begin{equation}\label{eq:last}
\lim_{l\to\infty}\|\psi_i^{(n_l)}-\psi_i\|\leq\|(1-\hat P)\psi_i\|\leq \|\psi_i\|.
\end{equation}
which implies that $|\psi_i\rangle\neq 0$. Otherwise, one would get
\begin{equation}
0=\lim_{l\to\infty}\|\psi_i^{(n_l)}-\psi_i\|=\lim_{l\to\infty}\|\psi_i^{(n_l)}\|,
\end{equation} 
which is contradictory to $\|\psi_i^{(n)}\|=1$. Therefore, according to \eqref{eq:psikeigen}, $\lambda_i$ is an eigenvalue of $\overline{\hpm}$ and $|\psi_i\rangle$ is a corresponding eigenvector.
\end{proof}
The above proof is inspired by the works \cite{levitin2004spectral,lewin2009spectral} and \cite[Section VIII.7]{reed2003methods}. 
 \begin{thm}\label{thm:inverse}
Given $\lambda_i$ and $\lambda_{i+1}$ as defined in \eqref{eq:lambdai},  if $\lambda_i\neq\lambda_{i+1}$, i.e. $\lambda_i<\lambda_{i+1}$, one has
\begin{equation}
\sigma(\overline{\hpm})\cap(\lambda_i,\lambda_{i+1})=\emptyset.
\end{equation}
 \end{thm}
 \begin{proof}
Consider an interval $(a,b)$ with $\lambda_i<a<b<\lambda_{i+1}$. By definition of $\lambda_i$ and \eqref{eq:orderlambda}, there exists an integer $\tilde N$ such that 
 \begin{equation}
 \lambda_i^{(k)}\leq a,\ \forall k\geq \tilde N.
 \end{equation}
Then for an eigenvalue $\lambda_{i'}^{(k)}$ of $\hat P^{(k)}\hpm\hat P^{(k)}$ with $k\geq \tilde N$,
\begin{itemize}
\item[(i)] if $i'\leq i$, one has
 \begin{equation}
 \lambda_{i'}^{(k)}\leq \lambda_i^{(k)}\leq a.
 \end{equation}
 \item[(ii)] if $i'>i$, or equivalently $i'\geq i+1$, one has
 \begin{equation}
 \lambda_{i'}^{(k)}\geq \lambda_{i+1}^{(k)}\geq \lambda_i\geq b.
 \end{equation}
\end{itemize}
The above analysis indicates
\begin{equation}\label{eq:emptyinter}
\sigma(\hat P^{(k)}\hpm\hat P^{(k)})\cap(a,b)=\emptyset,\ \forall k\geq \tilde N
\end{equation} 
 where $\sigma(\hat P^{(k)}\hpm\hat P^{(k)})$ denote, as usual, the set of eigenvalues of  $\hat P^{(k)}\hpm\hat P^{(k)}$. Let $z$ be the complex number 
 \begin{equation}
 z=\frac{a+b}{2}+i\frac{a-b}{2}.
 \end{equation}
 Given $|\varphi\rangle\in\dmu(m)$, one has $|\psi\rangle:=(\hpm-z) |\varphi\rangle\in \dmu(m)$ by the expression of $\hpm$. Hence there exists a large integer $n\geq\tilde N$ such that 
 \begin{equation}\label{eq:projectinvariant}
 \hat P^{(n)}|\varphi\rangle=|\varphi\rangle,\ \hat P^{(n)}|\psi\rangle=|\psi\rangle,\ \hat P^{(n)}\hpm|\varphi\rangle=\hpm|\varphi\rangle.
 \end{equation}
 Then one has
 \begin{equation}
|\psi\rangle=(\hpm-z)|\varphi\rangle=(\hat P^{(n)}\hpm\hat P^{(n)}-z)|\varphi\rangle,
 \end{equation}
 which leads to 
 \begin{equation}
 \begin{aligned}
 (\hpm-z)^{-1}|\psi\rangle=(\hat P^{(n)}\hpm\hat P^{(n)}-z)^{-1}|\psi\rangle.
 \end{aligned}
 \end{equation}
 Because of $\hat P^{(n)} \psi=\psi$, i.e. $\psi\in\hmu^{(n)}$, we have
 \begin{equation}
 \|(\hat P^{(n)}\hpm\hat P^{(n)}-z)^{-1}\psi\|^2=\sum_{i'=1}^{2n+1}|(\lambda_{i'}^{(n)}-z)^{-1}|^2|\langle\psi_{i'}^{(n)}|\psi\rangle|^2
 \end{equation}
 where $\psi_{i'}^{(n)}$ is the normalized eigenvector of $\hat P^{(n)}\hpm\hat P^{(n)}$ corresponding to the eigenvalue $\lambda_{i'}^{(n)}$.  
 According to the inequality $\lambda_1^{(n)}\leq \lambda_2^{(n)}\leq\cdots\leq \lambda_k^{(n)}\leq a\leq b\leq \lambda_{k+1}^{(n)}\leq\cdots\leq\lambda_n^{(n)}$, we have that
 \begin{equation}
 |(\lambda_{i'}^{(n)}-z)^{-1}|^2\leq \frac{2}{(a-b)^2}
 \end{equation}
 which implies
 \begin{equation}
 \|(\hpm-z)^{-1}\psi\|^2\leq \frac{2}{(a-b)^2}\sum_{i'=1}^{2n+1}|\langle\psi_{i'}^{(n)}|\psi\rangle|^2= \frac{2}{(a-b)^2}\|\psi\|^2,\ \forall\, |\psi\rangle \in\dmu(m). 
 \end{equation}
 Thus, one has
 \begin{equation}
 \|(\overline{\hpm}-z)^{-1}\psi\|^2\leq \frac{2}{(a-b)^2}\|\psi\|^2,\ \forall |\psi\rangle\in\overline{\dmu(m)}. 
 \end{equation}
As a consequence, 
\begin{equation}\label{eq:spectralradiu}
\rho((\overline{\hpm}-z)^{-1})\leq  \frac{\sqrt{2}}{b-a}
\end{equation}
where $\rho((\overline{\hpm}-z)^{-1})$ is the spectral radius of $(\overline{\hpm}-z)^{-1}$.  Because of the self-adjointness of $\overline{\hpm}$, the spectrum of $(\overline{\hpm}-z)^{-1}$ is
\begin{equation}
\sigma((\overline{\hpm}-z)^{-1})=\{(\lambda-z)^{-1},\lambda\in\sigma(\hpm)\}.
\end{equation}
Therefore, \eqref{eq:spectralradiu} leads to
\begin{equation}
(a,b)\cap\sigma(\hpm)=\emptyset.
\end{equation}
which is true for any $(a,b)\subset (\lambda_i,\lambda_{i+1})$. Hence, one obtains 
\begin{equation}
(\lambda_i,\lambda_{i+1})\cap\sigma(\hpm)=\emptyset.
\end{equation}
 \end{proof}

\section{The effective dynamics}\label{app:dynmceff}
The effective Hamiltonian constraint reads
\begin{equation}
H=p_\varphi^2-\frac{4\pi}{G L_0^2\gamma^2}\left(\frac{2}{\delta_b\delta_c} p_b\sin(\delta_b b) p_c\sin(\delta_c c) +\frac{1}{\delta_b^2}p_b^2\sin^2(\delta_b b)+\gamma^2p_b^2\right)=:p_\varphi^2-\frac{4\pi}{GL_0^2\gamma^2}\h.
\end{equation}
As $p_c\sin(\delta_c c)=\gamma m L_0\delta_c$ is a constant of motion, it is sufficient to consider the following Hamiltonian constraint for the evolution of $p_b$, 
\begin{equation}\label{eq:F2}
H^{(m)}=p_\varphi^2-\frac{4\pi}{GL_0^2\gamma^2}\left(\frac{2 \gamma m L_0}{\delta_b}p_b\sin(\delta_b b)+\frac{1}{\delta_b^2}p_b^2\sin^2(\delta_b b)+\gamma^2 p_b^2\right)=:p_\varphi^2-\frac{4\pi}{GL_0^2\gamma^2}\h^{(m)}.
\end{equation}
 The
evolution of $y:=p_b\sin(\delta_b b)$ with respect to $\varphi$ is given by
\begin{equation}\label{eq:evoy}
\begin{aligned}
&\frac{\dd y}{\dd \varphi}=\frac{\sqrt{4\pi}}{\sqrt{G} L_0 \gamma} \{y,\sqrt{\h^{(m)}}\}
=\frac{4\pi}{G L_0^2 \gamma^2}\frac{1}{p_\varphi}p_b^2\gamma^3\delta_b\cos(\delta_b b)=\pm\frac{4\pi \delta_b\gamma}{G L_0^2} \frac{1}{p_\varphi}\sqrt{p_b^4-p_b^2y^2}.
\end{aligned}
\end{equation}
The Hamiltonian constraint $H^{(m)}=0$ can also be written as
\begin{equation}\label{eq:pby}
p_b^2=-\frac{y^2}{\gamma^2\delta_b^2}-\frac{2 L_0 m y}{\gamma\delta_b}+\frac{G L_0^2 p_\varphi^2}{4\pi}
\end{equation}
which, together with \eqref{eq:evoy}, gives
\begin{equation}\label{eq:dydphi}
\frac{\dd y}{\dd \varphi}=\pm\frac{4\pi \sqrt{1+\gamma^2\delta_b^2}}{G L_0^2\gamma\delta_b} \frac{ 1}{p_\varphi}\sqrt{(y-y_-)(y-y_-')(y-y_+')(y-y_+)}
\end{equation}
with
\begin{equation}
\begin{aligned}
y_\pm&=L_0\gamma\delta_bm\left(-1\pm \sqrt{1+\frac{G p_\varphi^2}{4\pi m^2}}\right)\\
y_\pm'&=\frac{L_0\gamma\delta_b m}{1+\gamma^2\delta_b^2}\left(-1\pm\sqrt{1+\frac{G p_\varphi^2(1+\gamma^2\delta_b^2)}{4\pi m^2}}\right).
\end{aligned}
\end{equation}
Moreover, according to \eqref{eq:pby}, the maximal value of $p_b$ along its dynamical trajectory is
\begin{equation}
(p_b^{\rm max})^2=L_0^2m^2(1+\frac{G p_\varphi^2}{4\pi m^2})
\end{equation}
with which $y_\pm$ and $y_\pm'$ can be rewritten as
\begin{equation}\label{eq:F8}
\begin{aligned}
y_\pm&=L_0\gamma\delta_bm\left(-1\pm \frac{p_b^{\rm max}}{L_0m}\right)\\
y_\pm'&=\frac{L_0\gamma\delta_b m}{1+\gamma^2\delta_b^2}\left(-1\pm\sqrt{(1+\gamma^2\delta_b^2)\frac{(p_b^{\rm max})^2}{L_0^2m^2}-\gamma^2\delta_b^2}\right).
\end{aligned}
\end{equation}
Eq. \eqref{eq:F8} can be used to fix the dynamical solution by the initial data of $p_b^{\rm max}$.

Solution to \eqref{eq:dydphi} is 
\begin{equation}\label{eq:int}
\varphi=\pm\frac{G L_0^2\gamma\delta_b}{4\pi \sqrt{1+\gamma^2\delta_b^2}} p_\varphi \int_{y_a}^{y_b}\frac{\dd y}{\sqrt{(y-y_-)(y-y_-')(y-y_+')(y-y_+)}}
\end{equation}
As $p_b^2=-(y-y_-)(y-y_+)\geq 0$ and $y_-\leq y_-'\leq y_+'\leq y_+$, we have
\begin{equation}\label{eq:yrangle}
y_a\in [y_-',y_+'],\ y_b\in [y_-',y_+'].
\end{equation}
Moreover,  because \eqref{eq:pby} can be written as
\begin{equation}
p_b^2=-\frac{1}{\gamma^2\delta_b^2}(y-y_-)(y-y_+),
\end{equation}
\eqref{eq:yrangle} implies that $p_b$ cannot reach $0$ and will bounce at
\begin{equation}
p_b^{(\pm)}=-(y_\pm'-y_-)(y_\pm'-y_+).
\end{equation}

In order to calculate the integral \eqref{eq:int}, we define
\begin{equation}
t(y)=\sqrt{\frac{y_+'-y_-}{y_+'-y_-'}}\sqrt{\frac{y-y_-'}{y-y_-}}.
\end{equation}
Then \eqref{eq:int} becomes 
\begin{equation}
\begin{aligned}
\varphi=\pm \frac{G L_0^2\gamma\delta_b}{2\pi \sqrt{1+\gamma^2\delta_b^2}} p_\varphi\sqrt{\frac{1}{(y_+'-y_-)(y_+-y_-')}} \int_{t_a}^{t_b}\frac{\dd t}{\sqrt{(1-t^2)(1-k^2 t^2)}}.
\end{aligned}
\end{equation}
By choosing the initial data $\varphi_0=\varphi(y_-')$, we finally have
\begin{equation}
\varphi(y)-\varphi_0=\pm\frac{G L_0^2\gamma\delta_b}{2\pi \sqrt{1+\gamma^2\delta_b^2}} p_\varphi\sqrt{\frac{1}{(y_+'-y_-)(y_+-y_-')}}  F(\arcsin(t(y))|k),
\end{equation}
where $F(x|k)$ is the elliptic integral	of the first type.

\end{document}